\documentclass{article}

\usepackage[hidelinks]{hyperref}
\usepackage{url}
\usepackage{booktabs}
\usepackage{amsfonts}
\usepackage{nicefrac}
\usepackage{microtype}
\usepackage[dvipsnames]{xcolor}
\usepackage{makecell}

\usepackage{soul}
\usepackage{amsmath}
\usepackage{amsthm}
\usepackage{amssymb,mathtools}
\usepackage{bm}
\usepackage[inline]{enumitem}

\usepackage{breqn}
\usepackage{bbm}
\usepackage{graphicx}
\usepackage{multirow}

\usepackage{algorithm}
\usepackage{algpseudocode}
\usepackage{amsmath}

\usepackage{authblk}
\usepackage{fullpage}

\newcommand{\defeq}{\vcentcolon=}
\def\opm{\bar{P}}
\def\secBest#1{\textbf{\textcolor{Gray}{#1}}}
\DeclareMathOperator*{\argmax}{arg\,max}
\DeclareMathOperator{\pre}{pre}
\DeclareMathOperator{\cov}{cov}

\newtheorem{proposition}{Proposition}

\theoremstyle{remark}
\newtheorem{remark}{Remark}

\usepackage[round]{natbib}

\begin{document}

\title{Improving Active Learning with a Bayesian Representation of Epistemic Uncertainty}

\author[1]{Jake Thomas}
\author[2]{Jeremie Houssineau}
\affil[1]{Warwick Mathematics Institute, University of Warwick, UK.}
\affil[2]{Division of Mathematical Sciences, Nanyang Technological University, Singapore.}

\date{}

\maketitle

\begin{abstract}
A popular strategy for active learning is to specifically target a reduction in epistemic uncertainty, since aleatoric uncertainty is often considered as being intrinsic to the system of interest and therefore not reducible. Yet, distinguishing these two types of uncertainty remains challenging and there is no single strategy that consistently outperforms the others. We propose to use a particular combination of probability and possibility theories, with the aim of using the latter to specifically represent epistemic uncertainty, and we show how this combination leads to new active learning strategies that have desirable properties. In order to demonstrate the efficiency of these strategies in non-trivial settings, we introduce the notion of a possibilistic Gaussian process (GP) and consider GP-based multiclass and binary classification problems, for which the proposed methods display a strong performance for both simulated and real datasets.
\end{abstract}

\section{Introduction}

Active learning is a supervised learning problem in which the learner chooses which data points to learn from, with the goal of minimising the number of data points required to obtain an efficient model. In the sequential problem, at each time step, the learner is presented with a set of unlabelled points. Using the previously observed data, the learner must determine the most informative unlabelled point to query. Active learning strategies often aim to reduce epistemic uncertainty (EU), which is due to a lack of knowledge, rather than aleatoric uncertainty (AU), which is a fundamental aspect of the system that cannot be minimised. Distinguishing between EU and AU is however a significant challenge with the standard approach where both are modelled with probability distributions. The difficulties arising when using probability theory to model EU are highlighted, e.g., by the paradoxes \citep{dawid1973marginalization} appearing when using improper priors to model the absence of information, or by the fact that the standard notions of entropy and mutual information do not satisfy certain desirable properties \citep{wimmer2023quantifying}. In this work, we propose to leverage another representation of uncertainty with the aim of better capturing EU.

Alternatives to the standard inference frameworks have often been studied under different forms since Fisher's own attempt in the early 20th century \citep{fisher1935fiducial}. Possibility theory is one such approach which was first studied in the 70s \citep{zadeh1978fuzzy}. Possibility theory assumes various forms, and we consider one \citep{de2001integration} that closely resembles probability theory to facilitate the reformulation of the statistical models on which we will apply active learning. We consider in particular the formalism introduced by \cite{houssineau2018parameter} as it allows for leveraging a possibilistic analogue of Bayesian inference which preserves many of the properties of the standard approach while enabling a mixed probabilistic/possibilistic modelling. 

In Bayesian active learning \citep{houlsby2011bayesian, rodrigues2014gaussian}, the inference is often carried out by using Gaussian processes (GPs) \citep{williams2006gaussian} as they provide a versatile approach which yields a full posterior distribution of model uncertainty across input space, hence allowing for easy interpretability and use of the model. However, what is usually modelled by a GP is exactly the part that will be replaced by possibility functions in this work, and there is no possibilistic version of the notion of GP in the literature. We therefore introduce such a concept and show that it behaves similarly to the standard notion of GP.

After discussing some related work in Section~\ref{sec:relatedWork}, we outline relevant aspects of possibility theory in Section~\ref{sec:possibilityTheory} and introduce two novel strategies for active learning in Section~\ref{sec:activeLearningWithPossibilityTheory}, along with a discussion of their properties. In Section~\ref{sec:possibilisticGaussianProcesses}, we propose possibilistic analogues of the concepts of precision, covariance, and GP, and prove that these definitions preserve many of the useful properties of their probabilistic equivalents. We also highlight the key differences from the probabilistic case, including important distinctions around the roles of covariance and precision functions. Equipped with all the required tools, we assess the performance of the proposed active learning strategies in Section~\ref{sec:experimentalResults}.

Our contributions are as follows:
\begin{enumerate}[nosep]
    \item We propose two new acquisition functions for Bayesian active learning, including a new notion of epistemic uncertainty for which advantageous properties are highlighted
    \item Based on new results on the behaviour of the covariance in possibility theory,  we introduce a novel notion of possibilistic GP, which enables our acquisition functions to be implemented for non-trivial GP classification problems
    \item We show that the proposed approach displays a strong performance on both simulated and real datasets, with at least one of our acquisition functions outperforming the baselines in each case.
\end{enumerate}

\section{Related Work}
\label{sec:relatedWork}

This paper builds on existing work in information-theoretic active learning, focusing on varying measures of uncertainty over models and acquisition functions. The standard measure of uncertainty in the probabilistic context is the entropy over model parameters $\theta$ given labeled data $\mathcal{D}$, as used in several approaches \citep{sebastiani2000maximum, mackay1992information}. The corresponding greedy acquisition function selects the unlabelled input $x$ that maximally decreases the expected entropy over the parameters, i.e., $\arg\max_x H(\theta \vert \mathcal{D}) - \mathbb{E}_{y \sim p(y \vert x, \mathcal{D})}\left[ H(\theta \vert y, x, \mathcal{D}) \right]$. This requires knowledge of the posterior model parameters and is commonly rearranged to the BALD (Bayesian Active Learning by Disagreement) objective via the joint conditional mutual information \citep{houlsby2011bayesian}, focusing on the difference in entropy in label space, $\arg\max_x H(y \vert x, \mathcal{D}) - \mathbb{E}_{\theta \sim p(\theta \vert x, \mathcal{D})}\left[ H(y \vert x, \theta) \right]$. Despite their ubiquity, the performance of these active learning methods is not always consistent, and neither of them clearly outperforms standard baselines.

Recent studies have shown that possibility theory can be applied to complex uncertainty representation and inference problems. For instance, \cite{houssineau2018smoothing} introduced a possibilistic analogue of the Kalman filter, and \cite{houssineau2021linear} defined a possibilistic version of spatial point processes. These results emphasise the feasibility of reformulating GPs and introducing new active learning strategies in the context of possibility theory. Related to this work, \cite{hullermeier2021aleatoric} provide a comprehensive review of AU and EU in machine learning, supporting the motivation for our mixed possibilistic-probabilistic approach. \cite{wimmer2023quantifying} define a list of desirable properties that should be satisfied by notions of total uncertainty, AU, and EU, and show that the usual entropy and mutual information do not satisfy some of these properties. This list of properties is extended by \cite{sale2023second} who propose a family of solutions based on distances between a given second-order probability distribution and instances of such distributions identified as having no EU, AU or total uncertainty. \cite{denoeux2023reasoning} explores reasoning with fuzzy and uncertain evidence using epistemic random fuzzy sets, which aligns with our approach of combining possibility and probability theories.
\cite{caprio2023novel} introduces a generalisation of Bayes' theorem when both the likelihood and prior are only known to be contained within a given class of probability distributions, which could lead to generalisations of our approach where the likelihood is assumed to be known.
\cite{nguyen2019epistemic,nguyen2022measure} propose EU sampling in active learning, using the likelihood ratio and degrees of support for classification. Our approach contrasts by allowing prior information about the model to be taken into account and by generalising to multi-class problems.

More recently, several active learning methods have been proposed based on different principles: Mean Objective Cost of Uncertainty (MOCU) \citep{zhao2021efficient, zhao2021uncertainty} aims to reduce the expected loss, Bayesian Estimate of Mean Proper Scores (BEMPS) \citep{tan2021diversity,tan2023bayesian} aims to increase a score, and Expected Predictive Information Gain (EPIG) \citep{smith2023prediction} aims to increase the information gain. The acquisition function of all these methods includes an expected value over inputs, which means that they target the global improvement brought by each new label, rather than focusing on the local picture like most traditional active learning techniques; we refer to these methods as ``global'', as opposed to standard acquisition functions, which are ``local''. In addition, these global methods must first update the global model with an expected observation at the current input, as opposed to standard acquisition functions, which are ``update-free''. In this work, our main objective is to show that possibility theory allows to simply and efficiently capture EU, and to provide evidence for this via active learning. In particular, we propose two local and update-free acquisition functions and show that these outperform acquisition functions with the same properties.

\section{Possibility Theory and OPMs}
\label{sec:possibilityTheory}

\subsection{Introduction and properties}

Possibility theory and probability theory can be combined to provide a method of modelling both AU and EU \citep{houssineau2018parameter}. By relaxing the additive condition of probability theory when describing EU, we obtain outer measures which we refer to as ``outer \emph{probability} measures'' (OPMs) for which standard operations like conditioning or marginalisation can be defined.

To illustrate the core ideas of possibility theory, we consider the problem of modelling the uncertainty on an unknown but otherwise fixed parameter, $\theta^\star$, using the analogue of a random variable known as an uncertain variable. Uncertain variables (denoted by bold symbols) are functions, $\bm{\theta}: \Omega \to \Theta$ from a sample space  $\Omega$ which contains all possible states of nature to the parameter space $\Theta$ of possible values for $\theta^\star$. Note that there exists one true state of nature $\omega^\star \in \Omega$ such that $\bm{\theta}(\omega^\star) = \theta^\star$. In this case, since the uncertainty is entirely epistemic, the corresponding OPMs are of the form $\opm(A) = \sup_{\theta \in A} f_{\bm{\theta}}(\theta)$ for any $A \subseteq \Theta$, where $f_{\bm{\theta}}:\Theta \to [0,1]$ is a function such that $\sup_{\theta \in \Theta} f_{\bm{\theta}}(\theta) = 1$ which we refer to as a ``possibility function''. As in probability theory, it will often be easier to handle possibility functions directly rather than the corresponding OPMs.

Possibility theory allows us to describe the information we have about $\theta^\star$. Hence, we use the notion of \emph{credibility} rather than probability when talking about events, e.g., $\opm(A)$ is the credibility of the event $\bm{\theta} \in A$. A credibility of 1 indicates that there is no information to suggest an event is not possible, and a credibility of 0 is used when an event is not possible given the current information. A probability distribution \emph{characterises} a random variable, but a possibility $f_{\bm{\theta}}$ only \emph{describes} an uncertain variable $\bm{\theta}$. This is because many other possibility functions could describe the same uncertain variable, reflecting different levels of information. A review of the standard approach to possibility theory is given by \cite{dubois2015possibility}, where different notions of independence and conditioning are routinely studied \citep{dubois2000possibility}. We consider the context of \emph{numerical} possibility theory \citep{de2001integration} which follows rules that are analogous to the ones used in probability theory.

A notion of expected value has to be defined for uncertain variables, and we consider one that corresponds to the mode, as motivated by \cite{houssineau2019elements}, defined as $\mathbb{E}^\star(\bm{\theta}) = \argmax_{\theta \in \Theta} f_{\bm{\theta}}(\theta)$,
which can be set-valued. In particular, it holds that $\mathbb{E}^\star(T(\bm{\theta})) = T(\mathbb{E}^\star(\bm{\theta}))$ for any map $T$ on $\Theta$, similarly to maximum likelihood estimates. More generally, $T(\bm{\theta})$ is described via the change of variable formula by
$f_{T(\bm{\theta})}(\theta') = \sup_{\theta : T(\theta) = \theta'} f_{\bm{\theta}}(\theta)$, with the assumption that $\sup \emptyset = 0$ for cases where there is no $\theta$ such that $T(\theta) = \theta'$.

Given uncertain variables $\bm{\theta}$ and $\bm{\varphi}$ with respective sets of possible values $\Theta$ and $\Phi$, we define a joint possibility function as $f_{\bm{\theta}\bm{\varphi}}:\Theta \times \Phi \to [0,1]$ such that $\sup_{\theta \in \Theta, \varphi \in \Phi} f_{\bm{\theta}\bm{\varphi}}(\theta,\varphi) = 1$. The corresponding marginal possibility function describing $\bm{\theta}$ is $f_{\bm{\theta}}(\theta)= \sup_{\varphi \in \Phi} f_{\bm{\theta} \bm{\varphi} }(\theta,\varphi)$. We say that $\bm{\theta}$ and $\bm{\varphi}$ are independent if there exist possibility functions $f_{\bm{\theta}}$ and $f_{\bm{\varphi}}$ on $\Theta$ and $\Phi$ respectively such that $f_{\bm{\theta} \bm{\varphi} }(\theta,\varphi)=f_{\bm{\theta}}(\theta)f_{\bm{\varphi}}(\varphi)$ for all $\theta \in \Theta , \varphi \in \Phi$. The conditional distribution of $\bm{\theta}$ given $\bm{\varphi}$ is $f_{\bm{\theta}\vert \bm{\varphi}}(\theta\vert \varphi) = f_{\bm{\theta} \bm{\varphi} }(\theta,\varphi) / f_{\bm{\varphi} }(\varphi)$.

\subsection{Gaussian Possibility Function}

The multivariate Gaussian possibility function with parameters $\mu \in \mathbb{R}^n$ and $\Sigma \in \mathbb{R}^n\times \mathbb{R}^n$ is defined as $\overline{\mathrm{N}}(\theta;\mu,\Sigma) = \exp{\left(-\frac{1}{2}(\theta-\mu)^\top \Sigma^{-1}(\theta-\mu)\right)}$, for any $\theta \in \mathbb{R}^n$, with $\Sigma$ positive definite. This is simply a re-normalised version of the probabilistic Gaussian, yet it is as fundamental within possibility theory as the Gaussian distribution is within probability theory. In particular, it appears in asymptotic results such as the central limit theorem for uncertain variables \citep{houssineau2019elements}.

Consider an uncertain variable $\bm{\theta}$ described by the normal possibility function $\overline{\mathrm{N}}(\mu,\Sigma)$. Decomposing the uncertain variable $\bm{\theta}$ into $\bm{\theta}= (\bm{\theta}_1, \bm{\theta}_2)^{\top}$, with the corresponding mean $\mu = (\mu_1,\mu_2)^{\top}$ and with $\Sigma_{ij}$, $i,j \in \{1,2\}$, the corresponding sub-matrices of $\Sigma$, the conditional possibility function describing $\bm{\theta}_1$ given $\bm{\theta}_2 =\theta_2$ is $\overline{\mathrm{N}}(\theta;\mu_{1|2},\Sigma_{1|2})$ with $\mu_{1|2} = \ \mathbb{E}^\star(\bm{\theta}_1\vert \bm{\theta}_2=\theta_2) = \mu_1 + \Sigma_{21}\Sigma_{11}^{-1}(\theta_2 -\mu_2)$ and  with $\Sigma_{1|2} = \Sigma_{11} - \Sigma_{12}\Sigma_{22}^{-1}\Sigma_{21}$. Most of the standard results regarding Gaussian distributions continue to hold for Gaussian possibility functions, e.g., the marginal possibility function describing a subset of components of $\bm{\theta}$ is still Gaussian. Some of the advantages of the Gaussian possibility function are that
\begin{enumerate*}[label=\roman*)]
    \item it can be defined on a bounded subset $S$ of $\mathbb{R}^d$ without re-normalisation as long as $\mu$ belongs to $S$, and
    \item when parameterising it by the precision matrix $\Gamma = \Sigma^{-1}$, it allows for $\Gamma$ to be positive semi-definite rather than positive definite, in particular, one can set $\Gamma = 0$ when there is complete absence of information.
\end{enumerate*}

\subsection{Inference with OPMs}

Using the concept of OPM, one can also model the uncertainty about quantities for which EU and AU coexist. For instance, if we consider a random variable $Y$ on a set $\mathsf{Y}$, and if the true parameter $\theta^{\star}$ of the parameterised probability distribution $p_{\theta}$ of $Y$ is unknown but deterministic, then one can model the information about $\theta^{\star}$ via a possibility function $f_{\bm{\theta}}$. The OPM jointly describing the uncertainty about $\bm{\theta}$ and $Y$ is
\begin{equation} \label{OPM}
\opm(A \times B) = \sup_{\theta \in A} f_{\bm{\theta}}(\theta) \int_B p_{\theta}(y) \mathrm{d}y,
\end{equation}
for any $A \subseteq \Theta$ and any (measurable) subset $B$ of $\mathsf{Y}$. Henceforth, it will be convenient to slightly abuse notations and view $p_{\theta}$ as a conditional probability distribution $p_Y(\cdot \mid \theta)$. If we observe a realisation $y$ of $Y$, then we can show that the information about $\bm{\theta}$ given $Y=y$ is also a possibility function of the form
\[
f_{\bm{\theta}|Y}(\theta \mid y) = \dfrac{p_Y(y \mid \theta) f_{\bm{\theta}}(\theta)}{\sup_{\theta' \in \Theta} p_Y(y \mid \theta') f_{\bm{\theta}}(\theta')}.
\]
We will borrow from the Bayesian nomenclature and, e.g., refer to both $f_{\bm{\theta}\vert \bm{\varphi}}$ and $f_{\bm{\theta}|Y}$ as \emph{posterior} possibility functions.

\section{Active Learning with Possibility Theory}
\label{sec:activeLearningWithPossibilityTheory}

In information-theoretic active learning, the objective is to choose the unlabelled inputs (or set of unlabeled inputs) which produces the largest decrease in the uncertainty, i.e., to maximise the information gain. Therefore, a measure of uncertainty is required. The standard measure of uncertainty in the probabilistic context is the Shannon entropy \citep{sebastiani2000maximum, houlsby2011bayesian}. However, the typical notion of Shannon entropy, and its additive decomposition into conditional entropy and mutual information as measures of AU and EU respectively, fails to satisfy several natural properties which we would expect from such measures. We consider the list of properties that should be satisfied by notions of EU, AU and total uncertainty, as proposed by \cite{wimmer2023quantifying}; we focus on EU and, denoting $\mathrm{EU}(f_{\bm{\theta}})$ a given notion of EU based on any possibility function $f_{\bm{\theta}}$, adapt the corresponding properties to our framework as follows:
\begin{enumerate}[label=A\arabic*, start=0, nosep]
  \item $\mathrm{EU}(f_{\bm{\theta}})$ is non-negative.
  \item $\mathrm{EU}(f_{\bm{\theta}})$ vanishes if $f_{\bm{\theta}}$ is the indicator of a point
  \item $\mathrm{EU}(f_{\bm{\theta}})$ is maximal if $f_{\bm{\theta}}$ is the uninformative possibility function (i.e., $f_{\bm{\theta}}(\theta) = 1$, $\forall \theta \in \Theta$)
  \item If $T : \Theta \to \Theta$ is a bijective mapping such that, for some $\alpha \in (0,1]$ and some given probability distribution $p$ on $\mathsf{Y}$, it holds that $p_Y(y \vert T(\theta)) = \alpha p_Y(y \vert \theta) + (1 - \alpha) p(y)$ for all $y \in \mathsf{Y}$, then $\mathrm{EU}(f_{T(\bm{\theta})}) = \alpha \mathrm{EU}(f_{\bm{\theta}})$
  \item If $f'_{\bm{\theta}}$ is strictly less informative than $f_{\bm{\theta}}$,\footnote{i.e., there exists a region $A \subset \Theta$ of positive volume such that $f'_{\bm{\theta}}(\theta) > f_{\bm{\theta}}(\theta)$ for any $\theta \in A$} then $\mathrm{EU}(f'_{\bm{\theta}}) \geq \mathrm{EU}(f_{\bm{\theta}})$ (weak version) or $\mathrm{EU}(f'_{\bm{\theta}}) > \mathrm{EU}(f_{\bm{\theta}})$ (strict version)
\end{enumerate}
Property A3 is the only one that has been significantly modified when compared to \cite{wimmer2023quantifying}, who focuses on the multinomial setting, i.e., where $\Theta$ is a simplex and where $p_Y(y \vert \theta) = \theta_y$. This is further motivated in Remark~\ref{rem:PropA3} in Appendix~\ref{proofs:AL}.

In this section, we introduce two strategies for active learning. The first one is based on a novel measure of EU while the second one is based on the notion of \emph{necessity}, which is a consequence of the sub-additivity of OPMs. We study which of the properties above are satisfied by the proposed strategies, and apply them to active learning. The proofs for the results in this section can be found in Appendix~\ref{proofs:AL}.

\subsection{Measures of Epistemic Uncertainty}

In possibility theory, our current knowledge is encoded within the OPM, $\opm$, which maps subsets of $\Theta$ and measurable subsets of $\mathsf{Y}$ into [0,1]. Taking the notation from \eqref{OPM}, we can think of the \emph{credibility} of an event, $\bm{\theta} = \theta'$ or $Y \in B$, as the value assigned to that event by the respective marginal
\begin{align*}
    \opm ( \{\theta'\} \times \mathsf{Y}) & = f_{\bm{\theta}}(\theta') \int p_Y( y \vert \theta) \mathrm{d}y =f_{\bm{\theta}}(\theta'),\\
    \opm (\Theta\times  B) & = \sup_{\theta \in \Theta} f_{\bm{\theta}}(\theta) \int_B p_Y(y \vert \theta) \mathrm{d} y \doteq \opm_Y(B).
\end{align*}
The marginal in $\Theta$, which simplifies to the possibility function $f_{\bm{\theta}}$, describes the credibility of $\bm{\theta} = \theta'$ given the available information. On the other hand, the marginal in $Y$ can be interpreted as a notion of maximum expected value of the probability of $Y \in B$ given the available information, which can be seen as the total mass assigned to the event $Y \in B$ by the marginal OPM $\opm_Y$.

To measure the uncertainty encoded by the OPM $\opm$ we can integrate the marginals $f_{\bm{\theta}}$ and $\opm_Y$ across their respective spaces, whenever such integrals are well defined. The first integral $U^{\Theta}_{f_{\bm{\theta}}} = \int_\Theta f_{\bm{\theta}}(\theta) \mathrm{d}\theta$ we call the EU on the parameter space as it directly captures the uncertainty modelled by the possibility function on $\Theta$. We subscript only the possibility function $f_{\bm{\theta}}$ in $U^{\Theta}_{f_{\bm{\theta}}}$, as this measure of EU does not depend on the other term in the OPM $\opm$. A similar approach has been used for control under uncertainty by \cite{chen2021observer}. The second integral is a novel notion of EU on $\mathsf{Y}$ defined as
\begin{equation*}
U^{\mathsf{Y}}_{\opm} = \int_{\mathsf{Y}} \sup_{\theta \in \Theta} \Big( f_{\bm{\theta}}(\theta) p_Y(y \vert \theta) \Big) \mathrm{d}y - 1.
\end{equation*}
In this case, the EU modelled by $f_{\bm{\theta}}$ is mapped to $\mathsf{Y}$ via $p_Y(\cdot \vert \theta)$ before being integrated out. Although $U^{\mathsf{Y}}_{\opm}$ is based on the conditional probability $p_Y(\cdot \vert \theta)$, it is not a measure of total uncertainty or AU since it vanishes when the true parameter $\theta^\star$ is known, as will be demonstrated in the following proposition. We will write $U^{\Theta}$ and $U^{\mathsf{Y}}$ when there is no ambiguity about the underlying OPM.

\begin{proposition}
\label{prop:propertiesMeasuresUncertainty}
The measure of uncertainty $U^{\Theta}$ satisfies Properties A0-A2, Property A3 in the multinomial setting, as well as the strict version of Property A4, and $U^{\mathsf{Y}}$ satisfies Properties A0-A3 as well as the weak version of Property A4.
\end{proposition}

The fact that $U^{\mathsf{Y}}$ only satisfies the weak version of Property A4 is discussed in Remark~\ref{rem:PropA4} in Appendix~\ref{proofs:AL}.

%%%
\subsection{Possibilistic Active Learning for Classification}
\label{sec:AlForClassification}

In classification, a set of labels $L$ and a set of inputs $\mathsf{X}$ are considered, and we assume that we are in the situation where we have already seen the labels of some set of inputs, defining the available data $\mathcal{D}$. We then consider an unlabelled input $x \in \mathsf{X}$ and denote by $Y$ the unknown label. We assume that the true distribution $p_Y(\cdot; \theta^{\star}_x)$ of the label $Y$ in $L$ is part of a parametric family $\{ p_Y(\cdot; \theta_x) : \theta_x \in \Theta \}$ of distributions on $L$. This setting matches with the type of OPM considered so far with $\mathsf{Y} = L$. This section explains how the measure of uncertainty $U^{L}$ can be used to define active learning acquisition functions by quantifying the current uncertainty at the unlabelled inputs.

For binary classification, we let $L$ be $\{-1, 1\}$ and set $p_Y(1; \theta_x) = \Phi(\theta_x) = 1 - p_Y(-1; \theta_x)$, with $\Phi(\cdot)$ the standard normal c.d.f., for any $\theta_x \in \Theta = \mathbb{R}$. We assume that the available information about $\theta^{\star}_x$ is encoded by a posterior possibility function describing $\bm{\theta}_x$ given $\mathcal{D}$, of the form $\overline{\mathrm{N}}(\mu_x, \sigma^2_x)$ for some $\mu_x \in \mathbb{R}$ and some $\sigma^2_x > 0$. The corresponding OPM is characterised by
\begin{equation*}
    \bar{P}_Y(B) = \sup_{\theta \in \mathbb{R}}\overline{\mathrm{N}}(\theta;\mu_x, \sigma^2_x) \sum_{l \in B} p_Y(l; \theta),
\end{equation*}
for any subset $B$ of $\{-1, 1\}$. The corresponding measure of uncertainty $U^L_{\mathrm{bin}}$ is expressed as
\begin{align*}
    U^L_{\mathrm{bin}} & = \sum_{l\in \{-1,1\}} \bar{P}_Y(\{l\}) - 1 \\
    & =  \sup_{\theta \in \mathbb{R}}\overline{\mathrm{N}}(\theta; \mu_x,\sigma^2_x) \Phi(\theta) %\\
    %& \qquad\qquad\qquad
    + \sup_{\theta \in \mathbb{R}}\overline{\mathrm{N}}(\theta; \mu_x,\sigma^2_x) (1 - \Phi(\theta)) - 1.
\end{align*}

For multiclass classification, we let $n$ be the size of $L$ and define $p_Y(\cdot; \theta_x)$ as the (logistic) softmax for any indexed family $\theta_x = \{ \theta_x^l \}_{l \in L}$ with $\theta_x^l \in \mathbb{R}$ for any $l \in L$. e denote by $\mathbb{R}^L$ the set of all such indexed families. We assume that the available information about $\theta^{\star}_x$ is encoded by a posterior possibility function describing $\bm{\theta}_x$ given $\mathcal{D}$, of the form $\prod_{l\in L} \overline{\mathrm{N}}(\mu_{l,x}, \sigma^2_{l,x})$ for some $\mu_{l,x} \in \mathbb{R}$ and some $\sigma^2_{l,x} > 0$, i.e., the uncertain variables $\bm{\theta}_x^l$, $l \in L$, are mutually independent. A similar calculation as in the binary case yields a measure of uncertainty $U^L_{\mathrm{multi}}$ defined by
\[
U^L_{\mathrm{multi}} = \sum_{l\in L} \sup_{\theta \in \mathbb{R}^L} p_Y(l;\theta) \prod_{l'\in L} \overline{\mathrm{N}}(\theta; \mu_{l',x}, \sigma^2_{l',x}) - 1.
\]

While the exact supremum of these distributions is not analytically solvable, it is well-suited for numerical optimization techniques. Due to the form of the sigmoid function, the total uncertainty is biased towards EU that induces classification uncertainty, echoing the focus of prediction oriented active learning strategies.

\subsection{Necessity of classification}

Another relevant quantity for active learning arises naturally when considering OPMs: if we consider the most likely label $l^* \in L$ for an input $x$, i.e., $l^* =\arg\max_{l \in L} p_Y(l; \mu_x)$ (with $\mu_x = \{\mu_{l,x}\}_{l \in L} \in \mathbb{R}^L$ in the multiclass case), we can compute the possibility and the \emph{necessity} of correct classification, i.e., of $l^*$ being indeed the most likely label. Formally, the notion of necessity can be defined as follows: the necessity $N(E)$ of an event $E$ is defined based on the possibility $\bar{P}(E^{\mathrm{c}})$ of the complementary event $E^{\mathrm{c}}$ as $N(E) = 1 - \bar{P}(E^{\mathrm{c}})$. The possibility of correct classification is not useful as it is always equal to $1$, since $l^*$ is the most likely label, however the necessity of correct classification quantifies the uncertainty around the decision boundary and can be used as an acquisition function for active learning. The rationale for this approach can be illustrated in the case of binary classification: we compare two inputs $x$ and $x'$, where the probability of the label $1$ is known to be in the intervals $[0.6, 1]$ and $[0.45,0.55]$, respectively; there is more EU about the input $x$, yet that EU will not affect the decision of selecting label $1$, whereas the smaller EU for the input $x'$ is consequential in terms of classification. This is captured by the necessity of correct classification which would be equal to $1$ for input $x$ and to $0$ for input $x'$. The two following propositions provide the expression of the necessity of correct classification for the two types of classification discussed in Section~\ref{sec:AlForClassification}.

\begin{proposition}
\label{prop:nec_binary}
In the context of binary classification with the standard normal c.d.f., the necessity of correct classification at input $x$ is $N_{\mathrm{bin}} = 1 - \overline{\mathrm{N}}(0; \mu_x, \sigma_x^2)$.
\end{proposition}

\begin{proposition}
\label{prop:nec_multi}
In the context of multiclass classification with a (logistic) softmax, the necessity of correct classification at input $x$ is $N_{\mathrm{multi}} = 1 - \max_{l \neq l^*} \overline{\mathrm{N}}(\mu_{l,x}; \mu_{l^*,x}, \sigma_{l^*,x}^2 + \sigma_{l,x}^2)$.
\end{proposition}

In both cases, active learning is carried out by querying the input for which the necessity of correct classification is the smallest, which corresponds to the input with most EU at the decision boundary. This strategy requires either no optimisation/integration or a very simple maximisation over the label space, and scales well with the size of $L$ as a consequence. The behaviour of $N_{\mathrm{bin}}$ and $N_{\mathrm{multi}}$ is the opposite of the one of $U^L_{\mathrm{bin}}$ and $U^L_{\mathrm{multi}}$, i.e., the former increase when the latter decrease. Simply considering $1-N_{\mathrm{bin}}$ and $1-N_{\mathrm{multi}}$ is sufficient to prove some desirable properties as shown in the following proposition.

\begin{proposition}
\label{prop:nec_properties}
The credibility $1-N_{\mathrm{bin}}$ satisfies Properties A0-A2 as well as the strict version of Property A4, and $1-N_{\mathrm{multi}}$ satisfies Properties A0-A2 as well as the weak version of Property A4.
\end{proposition}

What is missing so far is the mechanism by which the parameters of the involved normal possibility functions are calculated based on the available data $\mathcal{D}$. In the probabilistic setting, using a GP would be one tangible option for carrying out such a task. Since a possibilistic version of GPs has never been considered prior to this work, we introduce it in the following section.

\section{Possibilistic Gaussian Processes}
\label{sec:possibilisticGaussianProcesses}

We begin by formally defining the notion of covariance in possibility theory and proving that it retains many of the properties of the probabilistic version, before defining a possibilistic version of the notion of GP and demonstrating how it can be utilised for modelling and prediction. Given the significantly different definitions in the probabilistic and possibilistic cases, it is somewhat surprising that the posterior GP models prove to be equivalent. We will highlight where the differences and advantages appear within the interpretation of the possibilistic model. Although the properties detailed in this section closely mirror those observed in the probabilistic case, there are important differences in their derivations, which are detailed in Appendix~\ref{proofs:PGP}.

\subsection{Covariance}\label{covariance-section}

This section introduces the concept of precision and covariance for uncertain variables, which is necessary for the later introduction of GPs within possibility theory. This is the first time these properties have been formally defined for possibility functions. The definition is a natural extension of the previously described notion of variance in possibility theory \citep{houssineau2019elements} as $\mathbb{V}^\star(\bm{\theta}) = \mathbb{E}^\star\big(-\frac{d^2}{d\theta^2}\log f_{\bm{\theta}}(\bm{\theta})\big)^{-1}$, under the assumption that the possibility function  $f_{\bm{\theta}}$ is twice continuously differentiable and uni-modal. This definition can be naturally extended to a definition of covariance as follows. Given an uncertain variable $\bm{\theta} = (\bm{\theta}_1,\dots,\bm{\theta}_n)$ with joint possibility function $f_{\bm{\theta}}$ we define the precision between $\bm{\theta}_i$ and $\bm{\theta}_j$ to be 
\begin{equation}
    \pre^{\star}_{\bm{\theta}}(\bm{\theta}_i, \bm{\theta}_j) = \mathbb{E}^\star \left( -\frac{\partial^2}{\partial \theta_i \partial \theta_j } \log f_{\bm{\theta}}(\bm{\theta}) \right)
\end{equation}
The precision matrix of $\bm{\theta}$ is defined accordingly as $\pre^{\star}\left(\bm{\theta} \right) = \left[ \pre^\star_{\bm{\theta}}(\bm{\theta}_i,\bm{\theta}_j)\right]_{i,j=1}^n$. Because of the properties of the expected value $\mathbb{E}^{\star}$, the precision matrix $\pre^{\star}(\bm{\theta})$ can be seen to correspond to the notion of observed information.

As is standard, the covariance matrix of $\bm{\theta}$ is defined to be the inverse, when it exists, of the precision matrix $\cov^{\star}(\bm{\theta}) = (\pre^{\star} \left(\bm{\theta} \right))^{-1}$. Finally, the covariance $\cov^{\star} \left(\bm{\theta}_i, \bm{\theta}_j \right)$ between $\bm{\theta}_i,\bm{\theta}_j$ for all $i,j$ is defined to be $(i,j)$-th entry of the matrix $\cov^{\star}(\bm{\theta})$. The covariance $\cov^{\star}$ enjoys many of the properties of the standard notion of covariance

\begin{proposition}\label{cov properties} For all $i,j,k \in \lbrace 1,\ldots , n \rbrace$, $\alpha,\beta \in \mathbb{R} $ and $F \in \mathbb{R}^{N\times n}$, it holds that
\begin{align*}
    \cov^\star(\bm{\theta}_i,\bm{\theta}_i) & = \mathbb{V}^\star(\bm{\theta}_i), \\
    \cov^\star(\bm{\theta}_i+\bm{\theta}_j,\bm{\theta}_k) & = \cov^\star(\bm{\theta}_i,\bm{\theta}_k) + \cov^\star(\bm{\theta}_j,\bm{\theta}_k), \\
    \cov^\star(\bm{\theta}_i, \bm{\theta}_j+\bm{\theta}_k) & = \cov^\star(\bm{\theta}_i,\bm{\theta}_j) + \cov^\star(\bm{\theta}_i,\bm{\theta}_k), \\
    \cov^\star (\bm{\theta}_i + \alpha,\bm{\theta}_j + \beta) & = \cov^\star (\bm{\theta}_i,\bm{\theta}_j), \\
    \cov^\star(\bm{\theta}_i,\bm{\theta}_j) & = \cov^\star(\bm{\theta}_j,\bm{\theta}_i), \\
    \cov^\star(\alpha \bm{\theta}_i,\beta\bm{\theta}_j) & = \alpha\beta \cov^\star( \bm{\theta}_i,\bm{\theta}_j), \\
    \cov^\star(F\bm{\theta}) & = F^\top \cov^\star(\bm{\theta}) F.
\end{align*}
Further, if $\bm{\theta}_i$ and $\bm{\theta}_j$ are independent, then $\cov^\star (\bm{\theta}_i,\bm{\theta}_j) = 0$.
    
\end{proposition}

As an example, consider an uncertain variable $\bm{\theta}$ described by the normal possibility function $\overline{\mathrm{N}}(\mu,\Sigma)$. Mirroring the probabilistic case we have that $\mathbb{E}^\star(\bm{\theta})=\mu$ and $\cov^\star(\bm{\theta})=\Sigma$. Although the notions of expected value and variance differ between possibility and probability theory in general, the possibilistic Gaussian maintains many of the convenient properties of the probabilistic version since the mean and mode are equal for Gaussian distribution and since the observed information is the inverse of the covariance matrix in the considered setting.

Given a marginal possibility function $f_{\bm{\theta}_{1:k}}$, with $k$ less than $n$, deduced from the full possibility function $f_{\bm{\theta}_{1:n}}$, it holds that the covariance calculated from the former agrees with the one calculated from the latter between any two uncertain variables $\bm{\theta}_i$ and $\bm{\theta}_j$ with $i,j \in \lbrace 1,\dots, k\rbrace$. This result, which we will refer to as the \emph{consistency criterion}, justifies why we do not indicate the underlying vector $\bm{\theta}$ when writing the covariance between two elements as we do for the precision. The definitions of the expected value $\mathbb{E}^{\star}$ and of the covariance matrix mean that the Laplace approximation becomes a moment matching method in the context of possibility theory.

\subsection{Definition of the Possibilistic Gaussian Process}
A possibilistic GP (PGP) on a set $\mathsf{X}$ is defined to be a collection of uncertain variables, any finite number of which are described by a joint Gaussian possibility function. The corresponding uncertain variable is a function $\bm{g}(\cdot)$ on $\mathsf{X}$ described by a possibility function which we will denote by $\overline{\mathrm{GP}}(m(\cdot), k(\cdot,\cdot))$, with $m$ the underlying expected function, i.e.\ $m(x) = \mathbb{E}^\star(\bm{g}(x))$ for any $x \in \mathsf{X}$, and $k$ the covariance function corresponding to the notion of covariance defined in Section~\ref{covariance-section}.  Note that the consistency requirement, or marginalisation property, is maintained in PGPs. As opposed to standard GPs, the PGP $\bm{g}$ models the possible values of a true and fixed function, say $g^{\star}$, rather than a distribution of random functions. For a given point $x \in \mathsf{X}$, we see that $\bm{g}(x)$ is an uncertain variable describing the information we have on the value of $g^{\star}$ at the point $x$.

\subsection{Modelling and Prediction with Possibilistic Gaussian Processes}

The separation of AU and EU within the considered framework enables us to consider both systematic and random sources of error when modelling with PGPs. Here we present four cases, regression with both random and systematic sources of error, and classification with random error in labelling in the binary and multiclass cases. In each case we consider training points $X = (X_1,\dots,X_n)$ with corresponding observed values $y = (y_1,\dots,y_n)^{\top}$, we aim to predict the mean $\mu_{\mathrm{t}}$, covariance $\Sigma_{\mathrm{t}}$, and when appropriate label probabilities at test points $X_{\mathrm{t}}$. The following propositions formalise the results for the four considered cases and presents the appropriate predictive equations. For brevity, the new notation presented in the propositions below is defined in the proofs in Appendix~\ref{proofs:PGP}.

\begin{proposition}[Regression with Systematic/Random Error]
\label{prop:GP pred regression}
Let observations be either of the form $\bm{y} = \bm{g}(X) + \bm{\epsilon}$ with $\bm{\epsilon}$ described by $\overline{\mathrm{N}}(0,\sigma^2I)$ or of the form $\bm{y} = \bm{g}(X) + \epsilon$ with $\epsilon$ characterised by $\mathrm{N}(0,\sigma^2I)$. Given a covariance function $K(\cdot,\cdot)$ satisfying the standard assumptions, the conditional possibility function for $\bm{g}_{\mathrm{t}}$ given $\bm{y}$ is found to be $\overline{\mathrm{N}}(\mu_{\mathrm{t}}, \Sigma_{\mathrm{t}})$ where
\begin{subequations}
\label{Reg with Sys Error-Prediction }
\begin{align} 
     \mu_{\mathrm{t}} &= K_{\mathrm{t}}[K+\sigma^2 I]^{-1}y \\
     \Sigma_{\mathrm{t}} &= K_{\mathrm{t}\mathrm{t}}-K_{\mathrm{t}}[K+\sigma^2 I]^{-1}K_{\mathrm{t}}^\intercal,
\end{align}
\end{subequations}
with $K=K(X,X)$, $K_{\mathrm{t}}=K(X_{\mathrm{t}},X)$ and $K_{\mathrm{t}\mathrm{t}}=K(X_{\mathrm{t}},X_{\mathrm{t}})$.
\end{proposition}

The following propositions deal with classification, for which we focus on the case of a random error. We follow the standard approach of discriminative GP classification and recover the standard results for our PGP by using the Laplace approximation, as discussed in Appendix~\ref{sec: Laplace Appendix}.

\begin{proposition}[Binary classification]
\label{prop:GP pred binary classification}
Let $L = \{-1, 1\}$ be the set of classes and let the probability of each class at input $X$ be determined by a sigmoid function $\sigma$ as $\sigma(\bm{g}(X))= p(Y=1|X)=1-p(Y=-1|X)$, with $\bm{g}$ an unknown latent function. Under a Laplace approximation, the posterior $f(g_{\mathrm{t}} \vert X,y, X_{\mathrm{t}})$ is approximated by the Gaussian possibility function $\overline{\mathrm{N}}(g_{\mathrm{t}}; \mu_{\mathrm{t}}, \Sigma_{\mathrm{t}})$ with parameters $\mu_{\mathrm{t}} = K_{\mathrm{t}}\nabla\log p(Y|\hat{g})$ and $\Sigma_{\mathrm{t}} = K_{\mathrm{t}\mathrm{t}}-K_{\mathrm{t}}[K+W^{-1}]^{-1}K_{\mathrm{t}}^\intercal$, with $\hat{g}$ the mode of the true posterior possibility function and with $W$ the negative Hessian of the log-likelihood.
\end{proposition}

In the following result, we use a generalised version of the shorthand notations introduced in Propositions~\ref{prop:GP pred regression} and \ref{prop:GP pred binary classification} for the kernel $K$ of multiple PGPs and for the parameter $\hat{g}$ in the Laplace approximation.

\begin{proposition}[Multiclass classification]
\label{prop:GP pred multiclass classification}
Let $L$ be the set of classes and let the probability of the classes in $L$ at input $X$ be characterised by the softmax of a set of latent function $\{\bm{g}_l\}_{l \in L}$, with $\bm{g}_l$ described by a PGP with kernel $K_l$ and with $\bm{g}_l$ independent of $\bm{g}_{l'}$ for any $l,l' \in L$ such that $l \neq l'$. Under a Laplace approximation, the posterior $f(g_{l,\mathrm{t}} \vert X,y, X_{\mathrm{t}})$ is approximated by the Gaussian possibility function $\overline{\mathrm{N}}(g_{l,\mathrm{t}}; \mu_{l,\mathrm{t}}, \Sigma_{l,\mathrm{t}})$ with parameters
$\mu_{l,\mathrm{t}} = (K_{l,\mathrm{t}})^\intercal (K_l)^{-1}\hat{g}_l$
and
$\Sigma_{l,\mathrm{t}}$ is the l-th block of $\Sigma_{\mathrm{t}} = \mathrm{diag}(k(X_t, X_t)) - Q_{\mathrm{t}} ^\top (K + W^{-1})^{-1} Q_{\mathrm{t}}
$
\end{proposition}

\section{Experimental Results}
\label{sec:experimentalResults}

In this section, we present an experimental evaluation of the proposed active learning strategies. Detailed experimental settings and additional figures are provided in Appendix~\ref{experimental-settings}, and code for reproducing the results is available at: [URL]. In all experiments, a hot-start of one sample for each class is provided, and one time step corresponds to adding one additional sample, as selected by the considered acquisition function.

\begin{table*}
\small
\setlength{\tabcolsep}{2pt}
\centering
\caption{Performance results, in terms of accuracy at the final time step, on real-world and synthetic datasets. Best and second-best performances are in black bold and grey bold respectively.}
\label{tab:results}
\begin{tabular}{c|c|c|c|c|c|c}
\hline
 & Sonar & Wine & Breast Cancer & Ionosphere & Fri & Vehicle \\
 \hline
  & med. \scriptsize [Q1-Q3] & med. \scriptsize [Q1-Q3] & med. \scriptsize [Q1-Q3] & med. \scriptsize [Q1-Q3] & med. \scriptsize [Q1-Q3] & med. \scriptsize [Q1-Q3] \\
 \hline
 Rand. & .788 \scriptsize [.763, .825] & .969 \scriptsize [.956, .975]  & .950 \scriptsize [.943, .953] & .780 \scriptsize [.763, .793] & .827 \scriptsize [.800, .850] & .890 \scriptsize [.866, .905] \\ 
 \hline
 Stand.\ & \secBest{.800} \scriptsize [.763, .828] & \textbf{.988} \scriptsize [\textbf{.986}, \textbf{.994}] &  \secBest{.957} \scriptsize \secBest{[.950, .960]} & .830 \scriptsize [.813, .847] & .873 \scriptsize [.859, .888] &  \textbf{.950} \scriptsize \textbf{[.941, .956]}\\
 \hline
 Max. Ent. & .675 \scriptsize [.625, .703] & .981 \scriptsize [.975, .988] &  .950 \scriptsize [.943, .953] & .802 \scriptsize [.786, .803] & .855 \scriptsize [.832, .868] &  .906 \scriptsize [.891, .919]\\
 \hline
 BALD & .788 \scriptsize [.763, .813] & .975 \scriptsize [.963, .981] &  .950 \scriptsize [.943, .953] & .810 \scriptsize [.790, .820] & .827 \scriptsize [.800, .847] & .893 \scriptsize [.876,  .911] \\
 \hline
 $U^L_{\mathrm{bin}}$ &  \textbf{.825} \scriptsize \textbf{[.800, .850]} &  \textbf{.988} \scriptsize [\secBest{.981}, \textbf{.994}] & .953 \scriptsize [.947, .957] & \textbf{.839} \scriptsize \textbf{[.821, .853]} & \secBest{.877} \scriptsize \secBest{[.861, .890]} & .931 \scriptsize [.921, .939] \\
 \hline
 $N_{\mathrm{bin}}$ & \secBest{.800} \scriptsize \secBest{[.775, .838]} & \textbf{.988}  \scriptsize [\secBest{.981}, .988] & \textbf{.960} \scriptsize \textbf{[.957,.963]} & \secBest{.837} \scriptsize \secBest{[.820, .850]} & \textbf{.886} \scriptsize \textbf{[.873, .895]} & \textbf{.950} \scriptsize [\secBest{.939},\textbf{.956}]\\
 \hline
\\
 \hline
 & Block in Center & Block in Corner & & Thyroid & Iris & Seismic \\
 \hline
 & med. \scriptsize [Q1-Q3] & med. \scriptsize [Q1-Q3] & & med. \scriptsize [Q1-Q3] &  med. \scriptsize [Q1-Q3] & med. \scriptsize [Q1-Q3]  \\
 \hline
 Rand. & .660 \scriptsize [.586, .781] & .758 \scriptsize [.611, .824] &
 Rand. & .905 \scriptsize [.873, .930] & .915 \scriptsize [.890, .940] & .878 \scriptsize [.848, .908] \\ 
 \hline
 BALD & \secBest{.925} \scriptsize \secBest{[.910, .940]} & \textbf{.930} \scriptsize [\secBest{.910}, \textbf{.941}] &
 Least Conf. & .930 \scriptsize [.920, \secBest{.950}] & \secBest{.940} \scriptsize [.930, \secBest{.958}] & \secBest{.915} \scriptsize [\secBest{.896}, \secBest{.929}] \\
 \hline
 Stand.\ & .920 \scriptsize [.905, .925] & \textbf{.930} \scriptsize \secBest{[.910, .940]} &
 Margin-based & .920 \scriptsize [.890, .930] & .910 \scriptsize [.883, .928] & .878 \scriptsize [.856, .899] \\
 \hline
 $U^L_{\mathrm{bin}}$ & \textbf{.935} \scriptsize \textbf{[.916, .955]} &  .923 \scriptsize [\textbf{.918}, .935] &
 $U^L_{\mathrm{multi}}$ & \textbf{.950} \scriptsize \textbf{[.930, .970]} & \textbf{.960} \scriptsize \textbf{[.945, .970]} & \textbf{.925} \scriptsize \textbf{[.904, .934]} \\
 \hline
 $N_{\mathrm{bin}}$ & \secBest{.925} \scriptsize [.905, .935] & \textbf{.930} \scriptsize [.900, \secBest{.940}] &
 $N_{\mathrm{multi}}$ & \secBest{.940} \scriptsize [\textbf{.930}, \secBest{.950}] & \secBest{.940} \scriptsize \secBest{[.933, .958]} & \secBest{.915} \scriptsize [.890, \secBest{.929}] \\
 \hline
  &  &  &Entropy & .935 \scriptsize [.913, \secBest{.950}] & .930 \scriptsize [.920, .950] & .910 \scriptsize [.890, .920]\\
\end{tabular}
\end{table*}

\begin{table}
\small
\setlength{\tabcolsep}{2pt}
\centering
\caption{Average rank in terms of accuracy on all binary-classification datasets. Best and second-best performances are in black bold and grey bold respectively. }
\label{tab:resultsRank}
\begin{tabular}{c|c|c|c|c|c|c}
& Rand. & Stand.\ & Max.\ Ent. & BALD & $U^L_{\mathrm{bin}}$ & $N_{\mathrm{bin}}$ \\
\hline
Final step & 5.12 & 2.12 & 4.50 & 3.75 & \secBest{2.00} & \textbf{1.37} \\
\hline
AUC & 5.25 & \secBest{2.25} & 3.50 & 5.04 & 2.29 & \textbf{1.95}
\end{tabular}
\end{table}

\paragraph{Baselines} To investigate the performance of the proposed strategies, we benchmark them against existing local and update-free acquisitions functions including: random, BALD, maximum entropy of the latent function (``Max.\ Ent.''), least confidence (``Least conf.''), margin-based, and entropy. For binary classification, the least-confidence, margin-based and entropy-based acquisition functions are equivalent and are simply referred to as the standard (``Stand.'')\ acquisition function. All the baselines use a standard GP as model. Further details of the baseline acquisition strategies are included in Appendix~\ref{sec: baseline strategies}.

\paragraph{Datasets} We test the proposed strategies on 11 datasets, including 6 binary-classification and 3 multiclass real-world datasets, which are available via the OpenML dataset repository \citep{vanschoren2014openml}, and 2 synthetic binary-classification datasets, similar to those in \cite{houlsby2011bayesian}. The first synthetic dataset is ``Block in Center'', which has a large dense block of uninformative points on the decision boundary, and the second one is ``Block in Corner'', which has a large block of uninformative points far from the decision boundary.

\paragraph{Results} The performance, in terms of accuracy at the final time step ($n=50$), is reported in Table~\ref{tab:results} for each method and each dataset. Full details of the datasets and experimental settings are available in Appendix~\ref{experimental-settings}. The binary and multiclass experiments were each run 200 and 50 times respectively. Each run of an experiment was completed with a different subset of the dataset as the unlabelled pool to make sure the true performance of the approaches on such datasets is captured. Table~\ref{tab:results} shows that both our acquisition functions are most often either best or second best, both in multiclass and binary classification problems. BALD performs well on the synthetic datasets, but its performance is subpar on real problems. To highlight that our acquisition functions perform consistently across episodes, and to diversify our performance metric, we also consider in Table~\ref{tab:resultsRank} the average rank both in terms of the accuracy at the final time step and of the area under the curve (AUC) in accuracy vs.\ time step. We chose the former as our primary metric since the latter can be significantly influenced by the behaviour in early time steps. The results in Table~\ref{tab:resultsRank} show that $N_{\mathrm{bin}}$ is most often the best in both metrics while $U^L_{\mathrm{bin}}$ is second best in terms of final accuracy and is close to the second best in terms of AUC. The acquisition functions $N_{\mathrm{bin}}$ is also particularly computationally efficient, being only 12\% slower to evaluate than the standard acquisition functions and 9\% faster than BALD.

 \section{Conclusion}
 \label{sec:conclusion}
 
In this paper, we have proposed two ways to specifically measure epistemic uncertainty by using a framework based on possibility theory. We have demonstrated that these measures of epistemic uncertainty exhibit favourable theoretical properties and shown how they can be leveraged for Bayesian active learning. Our empirical evaluations, conducted through the development of a novel possibilistic Gaussian process framework, demonstrated the efficacy of these strategies in both multiclass and binary classification tasks on a variety of synthetic and real-world datasets. This work contributes to the ongoing efforts in the field of active learning by providing novel techniques for uncertainty quantification and reduction. In future work, we hope to address several limitations of the proposed approach:
\begin{enumerate*}[label=\roman*)]
    \item to develop a method for calculating $U^{\mathsf{X}}$ when $\mathsf{X}$ is uncountable and analytically intractable, e.g., based on importance sampling,
    \item to generalise the necessity of the correct classification to other settings, and
    \item to develop additional Bayesian models in the possibilistic framework, to enable our approach to be applied in a broader variety of domains.
\end{enumerate*}

\bibliographystyle{apalike}
\bibliography{thebiblio}

\clearpage
\onecolumn

\appendix
\section{Properties of the Possibilistic Covariance} \label{prop of poss-covar}

This appendix contains the proofs of the statements in Proposition \ref{cov properties} and the consistency criterion as stated in Section~\ref{covariance-section}. 

We state Proposition~\ref{cov properties} once more here for ease of reference.

\setcounter{proposition}{4}
\begin{proposition} For all $i,j,k \in \lbrace 1,\ldots , n \rbrace$, $\alpha,\beta \in \mathbb{R} $ and $F \in \mathbb{R}^{N\times n}$, it holds that
\begin{subequations}
\begin{align}
    & \cov^\star(\bm{\theta}_i,\bm{\theta}_i) = \mathbb{V}^\star(\bm{\theta}_i)  \label{covar 1st} \\
    & \cov^\star(\bm{\theta}_i,\bm{\theta}_j) = \cov^\star(\bm{\theta}_j,\bm{\theta}_i) \label{covar 2nd} \\
    & \cov^\star (\bm{\theta}_i + \alpha,\bm{\theta}_j + \beta) = \cov^\star (\bm{\theta}_i,\bm{\theta}_j) \label{covar 3rd} \\
    & \cov^\star(\alpha \bm{\theta}_i,\beta\bm{\theta}_j)= \alpha\beta \cov^\star( \bm{\theta}_i,\bm{\theta}_j) \label{covar 4th} \\
    & \cov^\star(\bm{\theta}_i+\bm{\theta}_j,\bm{\theta}_k) =  \cov^\star(\bm{\theta}_i,\bm{\theta}_k) +  \cov^\star(\bm{\theta}_j,\bm{\theta}_k) \label{covar 5th} \\
    & \cov^\star(\bm{\theta}_i, \bm{\theta}_j+\bm{\theta}_k) = \cov^\star(\bm{\theta}_i,\bm{\theta}_j) + \cov^\star(\bm{\theta}_i,\bm{\theta}_k) \label{covar 6th} \\
    & \cov^\star(F\bm{\theta}) = F^\top \cov^\star(\bm{\theta}) F \label{covar 7th}
\end{align}
\end{subequations}
Further, if $\bm{\theta}_i$ and $\bm{\theta}_j$ are independently described, then $\cov^\star (\bm{\theta}_i,\bm{\theta}_j) = 0$.
\end{proposition}

Equations \ref{covar 1st}-\ref{covar 4th} follow immediately from the definition of covariance and variance, from the basic properties of derivatives and from the the consistency criterion.

\subsection{Proof of (\ref{covar 5th}) and (\ref{covar 6th})}\label{Proof Bilinear}

We assume $f_{\bm{\theta}}$ is twice continuously differentiable in each argument and uni-modal. Due to the consistency criterion and the symmetry stated in \eqref{covar 2nd} it is sufficient to consider the case of $\bm{\theta} = (\bm{\theta}_1,\bm{\theta}_2,\bm{\theta}_3)$ described by the possibility function $f_{\bm{\theta}}$. We define a new uncertain variable $\varphi = (\bm{\theta}_1 + \bm{\theta}_2,\bm{\theta}_3)$ which is described by the possibility function 
\[
    f_{\bm{\varphi}}(\theta_4,\theta_3)= \sup_{\theta_1 + \theta_2 = \theta_4}f_{\bm{\theta}}(\theta_1,\theta_2,\theta_3).
\]
Let $\mathbb{E}^\star(\bm{\theta}) \doteq \theta^\star \doteq (\theta_1^\star, \theta_2^\star,\theta_3^\star)$, so that the expected value of $\bm{\varphi}$ is $\mathbb{E}^\star(\bm{\varphi}) = (\theta_1^\star+ \theta_2^\star,\theta_3^\star)$. Further, define the components of the covariance matrix of $\bm{\theta}$ as
\[
\cov^\star(\bm{\theta}) \doteq
    \begin{bmatrix}
             c_{11} & c_{12} & c_{13}\\
             c_{12} & c_{22} & c_{23} \\
             c_{13} & c_{23} & c_{33}
         \end{bmatrix}.
\]

Since it holds that $f_{\bm{\theta}}(\theta^\star)=1$ and $f_{\bm{\theta}}^'(\theta^\star)=0$, the second-order Taylor expansion of $f_{\bm{\theta}}$ at $\theta^\star$ is
\[
    f_{\bm{\theta}}(\theta^\star + h) = 1+ h^\top \pre^\star(\bm{\theta})h + \mathcal{O}(||h||^3)
\]
with $h = (h_1,h_2,h_3)^{\top}$.

To find $f_{\bm{\varphi}}$ in terms of the Taylor expansion we consider $f_{\bm{\theta}}(\theta^\star + v)$ where $v = (h_1,h_4-h_1,h_3)$. Differentiating the Taylor expansion with respect to $h_1$, to find the critical points in terms of $(h_4,h_3)$, we obtain
\[
    h_1^c = \frac{%\splitfrac{
    -c_{11}c_{23}h_3 + c_{11}c_{33}h_4 + c_{12}c_{13}h_3 - c_{12}c_{23}h_3 +%}{
    c_{12}c_{33}h_4 - c_{13}^2 h_4 + c_{13}c_{22}h_3 - c_{13}c_{23}h_4}%}
    {c_{11}c_{33}+2c_{12}c_{33}-c_{13}^2-2c_{13}c_{23}+c_{22}c_{23}-c_{23}^2}.
\]
Now setting $w = (h_1^c,h_4-h_1^c,h_3)$ in the Taylor expansion we obtain
\[
    f_{\bm{\theta}}(\theta^\star + w) = 1+ w^\top \pre^\star(\bm{\theta})w + \mathcal{O}(||w||^3).
\]
This allows us to calculate the covariance matrix of $\varphi$ 
\[
\cov^\star(\bm{\varphi}) = 
        \begin{bmatrix}
             c_{11} + 2c_{12} +c_{22} & c_{13} +c_{23}\\
             c_{13} + c_{23} & c_{33} 
         \end{bmatrix},
\]
which, from the (1,2) entry, leads to
\[
\cov^\star(\bm{\theta}_1+\bm{\theta}_2,\bm{\theta}_3) = \cov^\star(\bm{\theta}_1,\bm{\theta}_3) + \cov^\star(\bm{\theta}_2,\bm{\theta}_3),
\]
as required.

\subsection{Proof of (\ref{covar 7th})}

Let $F \in \mathbb{R}^{N \times n}$ for any integer $N > 0$ and let $\bm{\varphi} = (\bm{\varphi}_1 , \bm{\varphi}_2,\dots,\bm{\varphi}_N)^{\top} = F\bm{\theta}$. Given that $\bm{\varphi}_i = \sum_{k=1}^n F_{ik} \bm{\theta}_k$ and using the bilinearity of the  covariance function (Equations \ref{covar 4th}, \ref{covar 5th} and \ref{covar 6th}), it follows that
\begin{align*}
    \cov^\star(\bm{\varphi}_i,\bm{\varphi}_j) &= \cov^\star\left(\sum_{k=1}^n F_{ik} \bm{\theta}_k,\sum_{l=1}^n F_{jl} \bm{\theta}_l\right)\\
    &= \sum_{k=1}^n F_{ik} \cov^\star\left(  \bm{\theta}_k,\sum_{l=1}^n F_{jl} \bm{\theta}_l\right)\\
    &= \sum_{k=1}^n \sum_{l=1}^n F_{ik} F_{jl} \cov^\star\left(  \bm{\theta}_k,  \bm{\theta}_l\right)\\
    &= \left( F \text{var}^\star(\bm{\theta})F^\top \right)_{ij}.
\end{align*}

\subsection{Proof of the consistency criterion}

To prove the consistency criterion, it is enough to consider the uncertain vectors $\bm{\theta}= (\bm{\theta}_1,\dots,\bm{\theta}_{n+1})$ and  $\bm{\theta}_{1:n}= (\bm{\theta}_1,\dots,\bm{\theta}_n)$ with marginal, with the latter being described by
\begin{equation} \label{marginal for consistency}
    f_{\bm{\theta}_{1:n}}(\theta_{1:n}) = \sup_{\theta_{n+1} \in \mathbb{R}} f_{\bm{\theta}}(\theta_{1:n},\theta_{n+1}).
\end{equation}
We approach the problem by calculating the marginal using the Taylor expansion of $f_{\bm{\theta}}$ about $\theta^\star$ as in Section~\ref{Proof Bilinear}, that is
\[
    f_{\bm{\theta}}(\theta^\star + h) = 1+ \frac{1}{2}h^\top \pre^\star(\bm{\theta})h + \mathcal{O}(||h||^3),
\]
for some $h \in \mathbb{R}^{n+1}$. The precision can be split into sub-matrices as
\begin{equation}
    \pre^\star(\bm{\theta}) =     \begin{bmatrix}
             P_{11} & P_{12} \\
             P_{21} & P_{22} \\
         \end{bmatrix},
\end{equation}
where $P_{11}$ is $n\times n$ matrix, $P_{12}=P_{21}^\top$ is vector in $\mathbb{R}^n$, and $P_{22}$ is a scalar. Differentiating the Taylor expansion with respect to $h_{n+1}$, the $n+1$-th component of $h$, we find that the supremum in the marginal of \eqref{marginal for consistency} is satisfied for
\[
    h_{n+1}^c = - P_{22}^{-1}h_{1:n}P_{12}.
\]
Defining
$h_c = \begin{bmatrix}
    h_{1:n} & h_{n+1}^c \\
\end{bmatrix}^\top$,
and substituting this equation back into the Taylor expansion, we obtain
\[
    f_{\bm{\theta}}(\theta^\star + h^c) = 1 +\frac{1}{2}h_c^\top \pre^\star(\bm{\theta}) h_c +  \mathcal{O}(||h_c||^3).
\]
Expanding the quadratic term in terms of $h_{1:n}$ gives
\[
    h_c^\top \pre^\star(\bm{\theta}) h_c = h_{1:n}^\top P_{11} h_{1:n} - P_{22}^{-1}h_{1:n}P_{12}P_{21}h_{1:n} %\\
    - h_{1:n}P_{12}P_{22}^{-1}h_{1:n}P_{21} + P_{22}^{-1}h_{1:n}P_{21}h_{1:n}P_{12}.
\]
Finally, differentiating the Taylor expansion gives the precision matrix of $\bm{\theta}_{1:n}$ as
\begin{equation}
    \pre^\star(\bm{\theta}_{1:n}) = P_{11} - P_{12}P_{22}^{-1}P_{21}.
\end{equation}
The covariance is simply the inverse of the precision. Using block matrix inverses the covariance is equal to the upper-left $n \times n$  block in $\cov^\star(\bm{\theta})$ as required.

\section{Proofs of results in Section~\ref{sec:possibilisticGaussianProcesses}}
\label{proofs:PGP}

Propositions~\ref{prop:GP pred regression}, \ref{prop:GP pred binary classification}, and \ref{prop:GP pred multiclass classification} are proved in Sections~\ref{sec:regression}, \ref{sec:BinaryClassification}, and \ref{sec:MulticlassClassification}, respectively.

\subsection{Regression}
\label{sec:regression}

\subsubsection{Systematic Error}

Consider a linear model $\bm{g}(x)= \phi(x)^\top \bm{w}$ where $\bm{w}$ is a predetermined weight vector which we aim to learn and where $\phi:\mathbb{R}^D \to \mathbb{R}^N $ maps $D$-dimensional inputs in the data space to $N$-dimensional feature space. The prior knowledge about $\bm{w}$ is modelled by a possibility function $f_{\bm{w}} = \overline{\mathrm{N}}(0, \Sigma_{\mathrm{p}})$. Additionally, we take into account the presence of some deterministic error/deviation $\bm{\epsilon}$ when measuring the value of $g$ at a given input point $X$, i.e., $\bm{y} = \bm{g}(X) + \bm{\epsilon}$, and we assume that $\bm{\epsilon}$ is described by $\overline{\mathrm{N}}(0,\sigma^2I)$. The source of these errors could be, e.g., inaccuracies in the measurement process due to physical limitations or numerical approximations rather than stochastic perturbations as in the standard GP approach (when ensuring the objectivity of the probabilistic representation). We can compute the expected function $m(\cdot)$ and covariance function $k(\cdot,\cdot)$ for the considered linear model as $\mathbb{E}^\star(\bm{y}_i) = \phi(x)^\top \mathbb{E}^\star(w) + \mathbb{E}^\star(\bm{\epsilon}_i) = 0$ and
\begin{align*}
    \cov^\star(\bm{y}_i,\bm{y}_j) & = \cov^\star(\phi(x_i)^\top\bm{w}+\bm{\varepsilon}_i,\phi(x_j)^\top\bm{w}+\bm{\varepsilon}_j) \\
    & = K(x_i,x_j)+\sigma^2\mathbbm{1}_{i=j},
\end{align*}
where $K(x_i,x_j) = \phi(x_i)^\top \cov^\star(\bm{w})\phi(x_j)$. In general settings, where we do not make assumptions about the structure of $\bm{g}$, any covariance function satisfying standard assumptions can be used.
The possibility function describing the function values $\bm{z} = \bm{g}(x)$ and $\bm{z}'= \bm{g}(x')$ at $x$ and $x'$, respectively, is
\begin{equation*}
    f_{\bm{z},\bm{z}'}(z,z') =
    \sup \left\lbrace f_{\bm{w}}(w) \bigg\vert w \in \mathbb{R}^N, \begin{bmatrix}
        z\\
        z'
    \end{bmatrix}=\begin{bmatrix}
        \phi^\top (x)\\
        \phi^\top (x')
    \end{bmatrix}w \right\rbrace.
\end{equation*}

We will consider training points $X = (X_1,\dots,X_n)$ with observed function values $y = (y_1,\dots,y_n)^{\top}$ which are realisations of the uncertain variables $\bm{y}_i = \bm{g}(X_i)+\bm{\epsilon}_i$, where $\bm{\epsilon}_1,\dots,\bm{\epsilon}_n$ are independently described. In addition, we have a test point $X_{\mathrm{t}}$ and the corresponding uncertain variable $\bm{y}_{\mathrm{t}}= \bm{g}(X_{\mathrm{t}}) + \bm{\epsilon}_{\mathrm{t}}$ in the feature space.

The joint prior possibility function describing the observed value $\bm{y}$ at the training points and $\bm{g}_{\mathrm{t}}$ at the test points is
  \[
     f_{\bm{y},\bm{g}_{\mathrm{t}}}(y,g_{\mathrm{t}}) =\overline{\mathrm{N}}\left(0,
         \begin{bmatrix}
             K+\sigma^2 I & K_{\mathrm{t}}^\intercal\\
             K_{\mathrm{t}} & K_{\mathrm{t}\mathrm{t}} \\
         \end{bmatrix}\right).
  \]

To simplify notation we set $ K= K(X,X)$, $K_{\mathrm{t}}=K(X_{\mathrm{t}},X)$ and $K_{\mathrm{t}\mathrm{t}}=K(X_{\mathrm{t}},X_{\mathrm{t}})$. The conditional possibility function for $\bm{g}_{\mathrm{t}}$ given $\bm{y}$ is found to be $\overline{\mathrm{N}}(\mu_{\mathrm{t}}, \Sigma_{\mathrm{t}})$ where
\begin{align} 
     \mu_{\mathrm{t}} &= K_{\mathrm{t}}[K+\sigma^2 I]^{-1}y \label{Reg with Sys Error-Prediction mean}\\
     \Sigma_{\mathrm{t}} &= K_{\mathrm{t}\mathrm{t}}-K_{\mathrm{t}}[K+\sigma^2 I]^{-1}K_{\mathrm{t}}^\intercal \label{Reg with Sys Error-Prediction covar}
\end{align}
Note the exact correspondence between the probabilistic and possibilistic cases despite the varying definitions.

%%%
\subsubsection{Random Error}\label{subsec:Reg w/ rand er}

In many situations the error when observing realisations of the function may be due to truly random effects. This situation can be effectively modelled with possibility theory. As before we consider training points $X = (X_1,\dots,X_n)$ with observed function values $y = (y_1,\dots,y_n)^{\top}$ which are realisations of the random variables $\bm{y}_i = \bm{g}(X_i)+\bm{\epsilon}_i$, where $\bm{\epsilon}_1,\dots,\bm{\epsilon}_n$ are independently distributed $\mathrm{N}(0,\sigma^2I)$. The possibility function describing $\bm{g}_t$ at test points $x_t$ is given by
\begin{equation} \label{random-regression}
    f(g_t|X,y, x_t) = \sup_{g \in G} f(g_t \vert X,X_t,g)f(g\vert X,y),
\end{equation}
where $f(g_t \vert X,X_t,g)$ is known from the zero-noise version of the previous section and
\begin{equation} \label{bayes-posterior}
    f(g\vert X,y) = \frac{p(y\vert X,g)  f(g \vert X)}{f(y\vert X)},
\end{equation}
where $f(y\vert X) = \sup_{g'} p(y \vert X,g) f(g \vert X )$. The conditional possibility function for $\bm{g}$ given $y$ is found to be $\overline{\mathrm{N}}(\mu, \Sigma)$ with $\mu = \sigma^{-2}\Sigma y$ and $\Sigma = (\sigma^{-2} I + K^{-1})^{-1}$.

The supremum in \eqref{random-regression} can then be solved using a formulation of the standard Kalman filter identity \cite{houssineau2018smoothing}. The resulting possibility function is identical to the systematic error case (see \eqref{Reg with Sys Error-Prediction mean} and \eqref{Reg with Sys Error-Prediction covar}), despite the noise being random in nature. 

%%%%%%
\subsection{Binary Classification}
\label{sec:BinaryClassification}

This section extends the PGP theory to binary classification problems, closely following the probabilistic treatment. We follow the standard approach of discriminative GP classification by considering a latent function, $g^\star$, which is passed through a sigmoid function, $\sigma(g^\star(X))= p(Y=1|X)=1-p(Y=-1|X)$, to determine the classification of each data point. Given data points $X = \lbrace x_1,x_2,\ldots x_n\rbrace$, with corresponding labels $Y=\lbrace y_1,\ldots, y_n \rbrace$, and test points $X_t$; the posterior distribution of the latent function is given by
\begin{equation}
    f(g_t|X,y, x_t) = \sup_{g \in G} f(g_t \vert X,X_t,g)f(g\vert X,Y). \label{classification-split}
\end{equation}
From the regression problem, $f(g_t \vert X,X_t,g)$ is known. However, as in the probabilistic case, the probit's, $p(Y\vert X,g)$, role within $f(g\vert X,Y)$ makes it analytically intractable. Therefore, we perform a Laplace approximation to obtain $q(g\vert X,Y) = \overline{\mathrm{N}}(g\vert K\nabla\log p(Y|\hat{g}), (K^{-1}+W^{-1})^{-1})$, details are contained within Appendix \ref{sec: Laplace Appendix}. This results in a normal approximation, $\overline{\mathrm{N}}(g_t\vert \mu, \Sigma)$, for $f(g_t \vert X,y, x_t)$, with
\begin{align}
     \mu &= K_{\mathrm{t}}\nabla\log p(Y|\hat{g}) \label{class posterior mean}\\
     \Sigma &= K_{\mathrm{t}\mathrm{t}}-K_{\mathrm{t}}[K+W^{-1}]^{-1}K_{\mathrm{t}}^\intercal. \label{class posterior covar}
\end{align}

The most credible probability assigned to class 1, is then given by $\mathbb{E}^\star(\sigma(\bm{g}_t))$, which is simply the probability corresponding to the most likely value of the PGP at the test point.

%%%%%%
\subsection{Multiclass Classification}
\label{sec:MulticlassClassification}

This section extends the PGP theory to multiclass classification problems, again we closely following the probabilistic treatment. Let $L$ be the set of classes and let the probability of the classes in $L$ at input $X$ be characterised the softmax of a set of latent functions $\{\bm{g}^l\}_{[l \in L]}$, with $\bm{g}^l$ described by a PGP with kernel $K_l$ and with $\bm{g}^l$ independent of $\bm{g}^{l'}$ for any $l,l' \in L$ such that $l \neq l'$.

Given data points $X = \lbrace x_1,x_2,\ldots x_n\rbrace$, with corresponding labels $Y=\lbrace y_1,\ldots, y_n \rbrace$, and test points $X_t$; the posterior distribution of the latent function is given by
\begin{equation}
    f(g_t|X,y, x_t) = \sup_{g \in G} f(g_t \vert X,X_t,g)f(g\vert X,Y). \label{multi classification-split}
\end{equation}
From the regression problem, $f(g_t \vert X,X_t,g)$ is known. However, as in the binary case, the softmax's, $p(Y\vert X,g)$, role within $f(g\vert X,Y)$ makes it analytically intractable. Therefore, we perform a Laplace approximation to obtain $q(g\vert X,Y) = \overline{\mathrm{N}}(K(y-\hat{\pi}), (K^{-1}+W^{-1})^{-1})$, details are contained within Appendix~\ref{sec: multiclass laplace appendix}

Thus the posterior $f(g_{\mathrm{t}} \vert X,y, X_{\mathrm{t}})$ is approximated by the Gaussian possibility function $\overline{\mathrm{N}}(g_{\mathrm{t}}; \mu_{\mathrm{t}}, \Sigma_{l,\mathrm{t}})$ with parameters
$\mu_{\mathrm{t}} = Q^\top_{\mathrm{t}}(y-\hat{\pi})$
and
$\Sigma_{\mathrm{t}} = \mathrm{diag}(k(X_t, X_t)) - Q_{\mathrm{t}} ^\top (K + W^{-1})^{-1} Q_{\mathrm{t}}
$. Where $k(X_t, X_t) = (K^1,\dots, K^{\vert L \vert})$ and we have utilised the notation
\[
    Q_t = \begin{pmatrix}
    K^1_{tt} & 0 & \cdots & 0 \\
    0 & K^2_{tt} & \cdots & 0 \\
    \vdots & \vdots & \ddots & \vdots \\
    0 & 0 & \cdots & K^{\vert L\vert}_{tt}
\end{pmatrix}.
\]

The most credible probability assigned to label $l$ at $X_i$, is then given by $\mathbb{E}^\star(\pi^l_i(\bm{g}_t))$, which is simply the probability corresponding to the most likely value of the PGP at the test point.

%%%%%%%%%%%%
\section{Laplace approximation} \label{sec: Laplace Appendix}

\subsection{Binary case}

We decompose the second term of equation \eqref{classification-split} with Bayes rule
\begin{equation} \label{bayes-posterior2}
    f(g\vert X,Y) = \frac{p(Y\vert X,g)  f(g \vert X)}{f(Y\vert X)}
\end{equation}
The terms in the numerator are known already as the sigmoid function and PGP prior. 
\begin{align*}
     &p(Y\vert X,g) = \sigma(g)\\
     & f(g \vert X) = \overline{\mathrm{N}}(0, K) \\
    f(Y\vert X) &= \sup_{g'} p(Y \vert X,g') f(g' \vert X )
\end{align*}

Due to the form of the sigmoid function, the posterior $f(g\vert X,Y)$ is non-Gaussian and therefore analytically intractable. We use the Laplace approximation on the posterior to find a Gaussian approximation $q(g\vert X,Y)$ for $f(g\vert X,Y)$. Performing a second-order Taylor expansion around the maximum of the posterior we obtain
\[
    q(g) = \overline{\mathrm{N}}(g\vert \Hat{g}, A^{-1})
\]

where $\Hat{g} = \argmax_{g}f(g\vert X,Y)$ and $A = -\nabla \nabla \log f(g\vert X,Y)$ is the hessian of the negative log posterior at this point. To find $\Hat{g}$ we need only consider the numerator of \eqref{bayes-posterior2}, as the denominator is independent of $g$. To this end, we define $\Psi$ to be the logarithm of the numerator and differentiate it twice 
\begin{align}
    \Psi (g) &\defeq  \log p(Y\vert X,g) + \log f(g\vert X)\\
    &= \log p(Y\vert X,g) -\frac{1}{2} g^\top K^{-1}g\\ \label{1stderivlaplace}
    \notag \\
    \nabla \Psi(g) &= \nabla\log p(Y\vert X,g) - K^{-1}g \\
   \nabla\nabla \Psi(g) &=  \nabla\nabla\log p(Y\vert X,g) - K^{-1} = -W -K^{-1} \label{2ndderivLaplace}
\end{align}

 $W\defeq -\nabla\nabla\log p(Y\vert X,g)$ is diagonal as the label at a point is only dependent on the value of $g$ at that point and nowhere else. And the location of the maximum is given by $\hat{g}=K\nabla\log p(Y\vert X,\hat{g}) $. As in the probabilistic approach, $\hat{g}$ cannot be found analytically but can be simply obtained using Newton's method. Using the covariance calculated in \eqref{2ndderivLaplace} we have the approximation 
 \begin{equation} \label{laplace-approx-final}
     q(g) = \overline{\mathrm{N}}(g;\hat{g}, (K^{-1}+W)^{-1})
 \end{equation}
Notably, the definitions of the expected value $\mathbb{E}^{\star}$ and covariance matrix mean that the Laplace approximation becomes a moment matching method within possibility theory.

%%%%%%
\subsection{Multiclass case} \label{sec: multiclass laplace appendix}

To proceed, we first introduce some additional notation:
\begin{align}
        g &= (g_1^1, \ldots, g_n^1, g_1^2, \ldots, g_n^2, \ldots, g_1^{\vert L \vert}, \ldots, g_n^{\vert L \vert})^\top \\
        g^l &= (g_1^1, \ldots, g_n^l)^\top\\
        g_i &= (g_i^1, \ldots, g_i^{\vert L \vert})^\top \\
        K^l &= K^l(X,X) \\
        K &= \text{Blockdiag}(\{K^l\}_{l \in L})
\end{align}
In addition, $y^l_i$ is one if $l$ is the class label for training point $i$ and 0 else, and $y=(y_1^1, \ldots, y_n^1, y_1^2, \ldots, y_n^2, \ldots, y_1^{\vert L \vert}, \ldots, y_n^{\vert L \vert})^\top$. Similarly, $\pi_i^l = p(y_i^l\vert g_i)$ and $\pi$ is defined accordingly

Following the binary case, we calculate multiclass analogue of \eqref{1stderivlaplace} and \eqref{2ndderivLaplace}
\begin{align}
        \nabla \Psi &= -\mathbf{K}^{-1} \mathbf{f} + \mathbf{y} - \pi \\
        \nabla\nabla \Psi(g) &= -W -K^{-1},
\end{align}
with
\[
\Psi(g) \defeq -\frac{1}{2}g^\top K^{-1}\mathbf{f} + y^\top g - \sum_{i=1}^n \log \left(\sum_{l=1}^{\vert L \vert} \exp(g_i^c) \right),
\]
where we have used \eqref{multi laplace trick} in calculating the second derivatives and $W \defeq \textit{diag}(\pi)-\Pi\Pi^\top$ with $\Pi$ the matrix obtained by stacking vertically the diagonal matrices
$\textit{diag}(\pi^l)$, and $\pi^l$ is the subvector of $\Pi$ relating to label l.
\begin{equation} \label{multi laplace trick}
    \frac{\partial^2}{\partial f_i^c \partial f_j^{c'}} \log \sum_j \exp(f_j^c) = \pi_i^c \delta_{cc'} - \pi_i^c \pi_i^{c'}
\end{equation}

Thus we have obtained the approximation $q(g\vert X,Y) = \overline{\mathrm{N}}(g;K(y-\hat{\pi}), (K^{-1}+W^{-1})^{-1})$, where $\hat{\pi} = p(y\vert \hat{g})$

%%%%%%%%%%%%
\section{Proofs of results in Section~\ref{sec:activeLearningWithPossibilityTheory}}
\label{proofs:AL}

\begin{remark}[On the reformulation of Property A3]
\label{rem:PropA3}
A more direct analogue of Property A3 as stated in \cite{wimmer2023quantifying} would be
\begin{enumerate}
\item[A3'] If $f'_{\bm{\theta}}$ is such that $f'_{\bm{\theta}}(\theta) = f_{\bm{\theta}}(\theta+c)$ for some constant $c \in \Theta$, then $\mathrm{EU}(f'_{\bm{\theta}}) = \mathrm{EU}(f_{\bm{\theta}})$
\end{enumerate}
which makes sense in the setting where $Y$ discrete, $f_{\bm{\theta}}$ a Dirichlet possibility function on the corresponding simplex $\Theta$, and with $p_Y(y \vert \theta) = \theta_y$, as considered in \cite{wimmer2023quantifying}. However, transformations of the form $f'_{\bm{\theta}}(\theta) = f_{\bm{\theta}}(\theta+c)$ are only well defined when the support of $f_{\bm{\theta}}$ verifies the following property: it holds that $\theta - c \in \Theta$ for any $\theta$ in the support of $f_{\bm{\theta}}$ (this applies equally to probability distributions). With the Dirichlet model, our version of Property A3 would be such that $T(\theta) = \alpha \theta + (1 - \alpha) c$, where the component of $c$ corresponding to label $y$ is equal to $p(y)$. This is similar but not equivalent to Property A3', which shows that Property A3 is not a generalisation of Property A3'. However, Property A3 allows for considering different forms for $p_Y(\cdot \vert \theta)$ and makes no assumption on the support of $f_{\bm{\theta}}$. The two measures of EU that we suggest satisfy both versions of this property, so our only motivation for introducing Property A3 is for its applicability beyond the Dirichlet model.
\end{remark}

\begin{proof}[Proof of Proposition~\ref{prop:propertiesMeasuresUncertainty}]
We start the proof with the notion $U^{\Theta}$ of EU on the parameter space $\Theta$.
\begin{enumerate}[label=A\arabic*, start=0]
  \item $U^{\Theta}$ is non-negative since possibility functions are non-negative.
  \item If $f_{\bm{\theta}}$ is the indicator $\bm{1}_{\hat{\theta}}$ of a point $\hat{\theta}$, that is $f_{\bm{\theta}}(\hat{\theta}) = 1$ and $f_{\bm{\theta}}(\theta) = 0$ if $\theta \neq \hat{\theta}$, then the integral of $f_{\bm{\theta}}$ is equal to $0$, so that $U^{\Theta} = 0$. Note that the indicator of a point is different from a Dirac delta which would give an integral of $1$.
  \item If $f_{\bm{\theta}}$ is the uninformative possibility function, that is $f_{\bm{\theta}}(\theta) = 1$ for any $\theta \in \Theta$, then it holds that $f_{\bm{\theta}}$ is pointwise greater than or equal to any other possibility function. It follows that $U^{\Theta}$ is indeed maximised.
  \item In the multinomial setting, we have $T(\theta) = \alpha \theta + (1 - \alpha) c$, so that a change variable in $U^{\Theta}_{f'_{\bm{\theta}}} = \int f'_{\bm{\theta}}(\theta) \mathrm{d}\theta$ suffices to show that $U^{\Theta}_{f'_{\bm{\theta}}} = \alpha U^{\Theta}_{f_{\bm{\theta}}}$.
  \item If there exists a region $A \subset \Theta$ of positive volume such that $f'_{\bm{\theta}}(\theta) > f_{\bm{\theta}}(\theta)$ for any $\theta \in A$ then
  \begin{align*}
      U^{\Theta}_{f'_{\bm{\theta}}} & = \int f'_{\bm{\theta}}(\theta) \mathrm{d}\theta \\
      & = \int_A f'_{\bm{\theta}}(\theta) \mathrm{d}\theta + \int_{\Theta \setminus A} f'_{\bm{\theta}}(\theta) \mathrm{d}\theta \\
      & > \int f_{\bm{\theta}}(\theta) \mathrm{d}\theta = U^{\Theta}_{f_{\bm{\theta}}}.
  \end{align*}
\end{enumerate}
We now consider the notion $U^{\mathsf{Y}}$ of EU on the set $\mathsf{Y}$, for which we prove the properties of interest in a different order. The conditional probability distribution $p_Y(\cdot \vert \theta)$ is the same throughout this proof but the underlying possibility function is modified, so we highlight the dependence on the latter by writing $U^{\mathsf{Y}}_{f_{\bm{\theta}}}$ when ambiguities arise.
\begin{enumerate}
  \item[A4] If there exists a region $A \subset \Theta$ of positive volume such that $f'_{\bm{\theta}}(\theta) > f_{\bm{\theta}}(\theta)$ for any $\theta \in A$ then
  \begin{align*}
      U^{\mathsf{Y}}_{f'_{\bm{\theta}}} & = \int \sup_{\theta \in \Theta} f'_{\bm{\theta}}(\theta) p_Y(y \vert \theta)  \mathrm{d} y - 1 \\
      & \geq \int \sup_{\theta \in \Theta} f_{\bm{\theta}}(\theta) p_Y(y \vert \theta) \mathrm{d} y - 1 = U^{\mathsf{Y}}_{f_{\bm{\theta}}}.
  \end{align*}
  \item[A1] If $f_{\bm{\theta}} = \bm{1}_{\hat{\theta}}$ for some $\hat{\theta} \in \Theta$, then
  \[
    U^{\mathsf{Y}}_{\bm{1}_{\hat{\theta}}} = \int_{\mathsf{Y}} p_Y(y \vert \theta) \mathrm{d}y - 1 = 0.
  \]
  \item[A2] If $f_{\bm{\theta}}$ is the uninformative possibility function, then it hold that $f'_{\bm{\theta}} \leq f_{\bm{\theta}}$ pointwise for any possibility function $f'_{\bm{\theta}}$. Therefore, by Property A4, $U^{\mathsf{Y}}_{f'_{\bm{\theta}}} \leq U^{\mathsf{Y}}_{f_{\bm{\theta}}}$ so that $U^{\mathsf{Y}}$ is indeed maximised for the the uninformative possibility function (but it might not be the unique maximiser).
  \item[A0] For every possibility function $f_{\bm{\theta}}$, there exists a point $\hat{\theta} \in \Theta$ such that $\bm{1}_{\hat{\theta}} \leq f_{\bm{\theta}}$ (in fact $\hat{\theta}$ is the expected value of $\bm{\theta}$ when $\bm{\theta}$ is described by $f_{\bm{\theta}}$). Therefore, by Property A4, $U^{\mathsf{Y}}_{f_{\bm{\theta}}} \geq U^{\mathsf{Y}}_{\bm{1}_{\hat{\theta}}} = 0$.
  \item[A3] Since $T$ is bijective, it holds that $f_{T(\bm{\theta})}(\theta') = f_{\bm{\theta}}(T^{-1}(\theta'))$ for any $\theta' \in \Theta$. It follows that
  \begin{align*}
    U^{\mathsf{Y}}_{f_{T(\bm{\theta}})} & = \int \sup_{\theta' \in \Theta} f_{\bm{\theta}}(T^{-1}(\theta') p_Y(y \vert \theta')  \mathrm{d} y - 1 \\
    & = \int \sup_{\theta \in \Theta} f_{\bm{\theta}}(\theta) p_Y(y \vert T(\theta))  \mathrm{d} y - 1 \\
    & = \alpha \int \sup_{\theta \in \Theta} f_{\bm{\theta}}(\theta) p_Y(y \vert T(\theta))  \mathrm{d} y + (1 - \alpha) - 1 \\
    & = \alpha U^{\mathsf{Y}}_{f_{\bm{\theta}}},
  \end{align*}
  where a change of variable in the supremum was considered between the first and second line.
\end{enumerate}
This concludes the proof of the proposition.
\end{proof}

\begin{remark}[On $U^{\mathsf{Y}}$ only satisfying the weak version of Property A4]
\label{rem:PropA4}
The notion $U^{\mathsf{Y}}$ of EU is defined on the set $\mathsf{Y}$, and it is therefore not surprising that it only satisfies the weak version of Property A4 since some information gains about $\bm{\theta}$ might not change the credibility of any point in $\mathsf{Y}$. For instance, in a binary classification case, i.e.\ $\mathsf{Y} = \{-1, 1\}$, $\Theta = [0,1]$ and $p_Y(1 \vert \theta) = \theta = 1 - p_Y(-1 \vert \theta)$, if it is known that the probability of the positive class is either $1$ or $0$, then each label could occur either all the time or never, which results in this case achieving the maximum value for $U^{\mathsf{Y}}$ despite the informativeness of the underlying possibility function $f_{\theta}$, equal to the indicator of the set $\{0,1\}$.
\end{remark}

In the following two propositions, we denote by $\bar{\mathbb{P}}$ the OPM on $\Omega$ induced by the considered possibility functions. This way, we can measure the credibility of any event $E \subseteq \Omega$ as $\bar{\mathbb{P}}(E)$.

\begin{proof}[Proof of Proposition~\ref{prop:nec_binary}]
Let $l^* = 1$ be the most likely label, i.e., $p_Y(1 ; \mu_x) > p_Y(-1 ; \mu_x)$, which is equivalent to $p_Y(1 ; \mu_x) > 1/2$. We also assume that the label $1$ is such that $p_Y(1 ; \theta) = \Phi(\theta)$. The possibility of the event $p_Y(1 ; \bm{\theta}) > 1/2$ is $1$ and the necessity of this event is
\begin{align*}
N_{\mathrm{bin}} & = 1 - \bar{\mathbb{P}}( p_Y(-1 ; \bm{\theta}) \geq 1/2 ) \\
& = 1 - \sup \Big\{ \overline{\mathrm{N}}(\theta; \mu_x, \sigma^2_x) : \theta\in \mathbb{R}, p_Y(-1 ; \theta) \geq 1/2 \Big\}.
\end{align*}
The condition $p_Y(-1 ; \theta) \geq 1/2$ can be expressed as $1 - \Phi(\theta) \geq 1/2$, which further simplifies to $\theta \leq 0$. Since $\mu_x \geq 0$, the maximum value of $\overline{\mathrm{N}}(\theta; \mu_x, \sigma^2_x)$ over the set of non-positive numbers is achieved at $\theta = 0$. Therefore, the necessity of interest is
\[
N_{\mathrm{bin}} = 1 - \overline{\mathrm{N}}(0; \mu_x, \sigma^2_x).
\]
The expression is the same if $-1$ is the most likely label.
\end{proof}

\begin{proof}[Proof of Proposition~\ref{prop:nec_multi}]
Assume that $l^*$ is the most likely label in $L$, then the necessity for $l^*$ to be the correct classification is
\[
N_{\mathrm{multi}} = 1 - \bar{\mathbb{P}}\big( \exists l \in L, l \neq l^*, p_Y(l ; \bm{\theta}_l) \geq p_Y(l^* ; \bm{\theta}_{l^*}) \big).
\]
Since $\bar{\mathbb{P}}( \cdot )$ refers to a possibility function in this case, we have that $\bar{\mathbb{P}}\big( \bigcup_{i=1}^n A_i \big) = \max_{i \in \{1,\dots,n\}} \bar{\mathbb{P}}( A_i )$. Indeed, there is no double counting with the maximum as there would be with a sum in the case of non-empty intersection between the sets $A_i$. Therefore, we obtain
\[
N_{\mathrm{multi}} = 1 - \max_{l \neq l^*} \bar{\mathbb{P}}\big( p_Y(l ; \bm{\theta}_l) \geq p_Y(l^* ; \bm{\theta}_{l^*}) \big).
\]
With $p_Y(\cdot ; \theta)$ being the (logistic) softmax, the condition $p_Y(l ; \bm{\theta}_l) \geq p_Y(l^* ; \bm{\theta}_{l^*})$ is equivalent to $\bm{\theta}_l \geq \bm{\theta}_{l^*}$. Similarly to the case of binary classification, since $\mu_{l^*,x} > \mu_{l,x}$, the possibility of the event $\bm{\theta}_l \geq \bm{\theta}_{l^*}$ is equal to the possibility of $\bm{\theta}_l = \bm{\theta}_{l^*}$. It follows that
\[
N_{\mathrm{multi}} = 1 - \max_{l \neq l^*} \max_{\theta \in \mathbb{R}} \overline{\mathrm{N}}(\theta; \mu_{l,x}, \sigma^2_{l,x})\overline{\mathrm{N}}(\theta; \mu_{l^*,x}, \sigma^2_{l^*,x}),
\]
which simplifies to
\[
N_{\mathrm{multi}} = 1 - \max_{l \neq l^*} \exp\bigg( -\frac{1}{2}\frac{(\mu_{l^*,x} - \mu_{l,x})^2}{\sigma_{l^*,x}^2 + \sigma_{l,x}^2}  \bigg),
\]
as required.
\end{proof}

\begin{proof}[Proof of Proposition~\ref{prop:nec_properties}]
For both versions of the necessity of correct classification, the quantities $1 - N_{\mathrm{multi}}$ and $1 - N_{\mathrm{bin}}$ are equal to the possibility $\bar{\mathbb{P}}\big( \exists l \in L, l \neq l^*, p_Y(l ; \bm{\theta}_l) \geq p_Y(l^* ; \bm{\theta}_{l^*}) \big)$, which is therefore in the interval $[0,1]$. Since $l^*$ is the unique maximiser of $p_Y( \cdot; \mu_x)$ and since $p_Y( l; \mu_x)$ is strictly monotonic as a function of $\mu_x$ (resp.\ $\mu_{l,x}$), it follows that $\mu_x \neq 0$ (resp.\ $\mu_{l^*,x} \neq \mu_{l,x}$ for all $l \neq l^*$).
\begin{enumerate}
  \item[A0] It holds that $1 - N_{\mathrm{bin}} \geq 0$ (resp.\ $1 - N_{\mathrm{multi}} \geq 0$) since possibility functions have a maximum equal to 1.
  \item[A1] If it holds that $\sigma_x \to 0$ (resp.\ $\sigma_{l,x} \to 0$ for all $l \in L$) then $\overline{\mathrm{N}}(0; \mu_x, \sigma^2_x) \to 0$ (resp.\ $\overline{\mathrm{N}}(\mu_{l,x}; \mu_{l^*,x}, \sigma_{l^*,x}^2 + \sigma_{l,x}^2) \to 0$ for all $l \neq l^*$).
  \item[A2] If it holds that $\sigma_x \to \infty$ (resp.\ $\sigma_{l,x} \to \infty$ for all $l \in L$) then $\overline{\mathrm{N}}(0; \mu_x, \sigma^2_x) \to 1$ (resp.\ $\overline{\mathrm{N}}(\mu_{l,x}; \mu_{l^*,x}, \sigma_{l^*,x}^2 + \sigma_{l,x}^2) \to 1$ for all $l \neq l^*$), which is the maximum value.
  \item[A4] We first deal with the binary case. Considering two Gaussian possibility functions $\overline{\mathrm{N}}( \mu_x, \sigma^2_x)$ and $\overline{\mathrm{N}}( \mu'_x, \sigma^{\prime 2}_x)$, the only way for the latter to be strictly less informative than the former is for $\mu'_x = \mu_x$ and $\sigma^{\prime 2}_x > \sigma^2_x$. Therefore, it holds that $\overline{\mathrm{N}}(0; \mu'_x, \sigma^{\prime 2}_x) > \overline{\mathrm{N}}(0; \mu_x, \sigma^2_x)$. If we denote by $N'_{\mathrm{bin}}$ the necessity associated with $\overline{\mathrm{N}}( \mu'_x, \sigma^{\prime 2}_x)$, we obtain that $1 - N'_{\mathrm{bin}} > 1 - N_{\mathrm{bin}}$ as required. Now considering the multiclass case, since we are dealing with product of normal possibility functions, it is enough for one of them to strictly less informative for the product to be strictly less informative. However, if this term is neither $l^*$ nor the one achieving the maximum over $l \neq l^*$, then $N_{\mathrm{multi}}$ is unchanged by such a loss of information, so that we can only obtain that $1 - N'_{\mathrm{multi}} \geq 1 - N_{\mathrm{multi}}$ in general.
\end{enumerate}
\end{proof}

%%%%%%%%%%%%
\section{Experiment Protocol}\label{experimental-settings}

This appendix details the experimental settings for the three classes of experiments. The code will be made available on GitHub upon acceptance. All experiments were conducted in a custom active learning pipeline implemented in Python. The pseudocode is available in Appendix \ref{sec: Pseudocode}. The GP models were implemented using a custom version of GPytorch v1.11 \cite{gardner2018gpytorch}, detailed in Appendix \ref{sec: gp models}. Details of the baseline active learning strategies are included in Appendix \ref{sec: baseline strategies}. All experiments were run on a laptop with an NVIDIA GeForce RTX 3080 Ti GPU. Each binary classification experiment took less than 2 hours to run, and each multiclass experiment took less than 8 hours to run. The optimisation for approach $U^L$ was solved using a Scipy's Brent Optimizer.

\subsection{Additional experimental details and results}

The training parameters used and references for the datasets are contained within Table~\ref{exper-settings-table}. To complement the results shown in table format in the main text, the accuracy for each method as a function of the time step is shown in Figure~\ref{realworld-classification}, where we can see that the considered datasets are of varying difficulty, hence allowing to thoroughly assess the performance of the proposed acquisition functions. In \texttt{Fri} and \texttt{Breast Cancer}, the acquisition function $N_{\mathrm{bin}}$ performs very well throughout training, whereas $U^L_{\mathrm{bin}}$ clearly outperforms all the other methods in \texttt{Sonar}. In \texttt{Wine}, \texttt{Vehicle} and \texttt{Ionosphere}, the difference between the 3 leading methods, i.e., $N_{\mathrm{bin}}$, $U^L_{\mathrm{bin}}$ and \emph{Standard}, is less pronounced. These results are not surprising as optimal behaviours differ between datasets, which is difficult to capture with local methods.

\begin{figure*}[t]
    \centering
    \includegraphics[width=\linewidth,trim=0pt 10pt 0pt 10pt,clip]{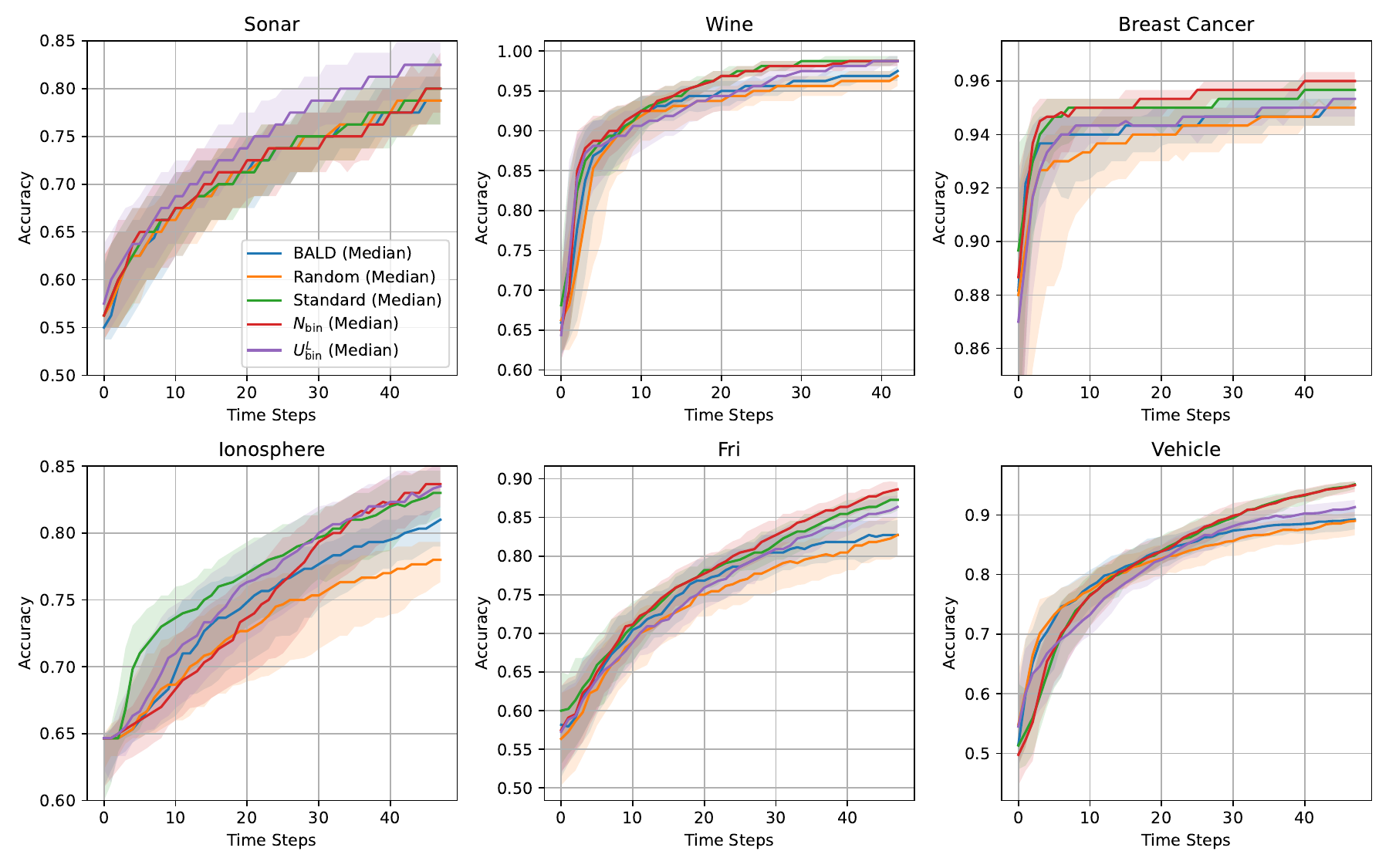}
    \caption{Performance of BALD, Random, Standard and our approaches on six binary classification problems on real-world datasets. The shaded area corresponds to the interval between the first quantile (Q1) and third quantile (Q3).}
    \label{realworld-classification}
\end{figure*}

\begin{table*}[!ht]
%\small
\footnotesize
\centering
\caption{This table contains all parameters for the experiments described in Section 6} \label{exper-settings-table}
\begin{tabular}{ |c|c|c|c|c|c| } 
 \hline
 Dataset & Reference & No.\ of runs & No.\ of queries & Size training & Size testing \\
 \hline
 Block in Center & \cite{houlsby2011bayesian} & 50 & 25 & 150 & 200 \\ 
 Block in Corner & \cite{houlsby2011bayesian} & 50 & 25 & 150 & 200 \\ 
 Sonar & \cite{gorman1988analysis}& 300 & 50 & 120 & 80 \\
 Wine & \cite{aeberhard1994comparative} & 200 & 45 & 120 & 160 \\
 Breast Cancer & \cite{mangasarian1990cancer} & 200 & 50 & 150 & 300 \\
 Ionosphere & \cite{sigillito1989classification} & 200 & 50 & 120 & 300\\
 Fri & \cite{friedman2002stochastic} & 200 & 50 & 120 & 220\\
 Vehicle & \cite{siebert1987vehicle} & 200 & 50 & 300 & 800 \\
 Thyroid & \cite{misc_thyroid_disease_102}  & 50 & 40 & 115 & 100 \\
 Iris &  \cite{misc_iris_53} & 50 & 40 & 100 & 100 \\
 Seismic & \cite{misc_seismic-bumps_266}  & 50 & 40 & 150 & 200 \\
 \hline
\end{tabular}
\end{table*}

\newpage

\subsection{Active Learning Pseudocode}\label{sec: Pseudocode}

\begin{algorithm}[!ht]
\caption{Active Learning Pipeline}
\begin{algorithmic}[1]
\Require Acquisition function $A$, Dataset $D$, Oracle $O$
\For{No. of Runs}
    \State $U \gets$ Randomly select \textit{Size of training set} points from $D$
    \State $L \gets$ Random sample of each label in $U$ labeled by $O$
    \State Remove these samples from $U$
    \For{No. of Queries}
        \State Train GP model on $L$
        \State $P \gets$ Predictions of GP model on $U$
        \State $d \gets A(P)$ \Comment{Select data point to query}
        \State $l \gets O(d)$ \Comment{Get label from oracle}
        \State $L \gets L \cup (d, l)$
        \State $U \gets U \setminus d$
    \EndFor
    \State $T \gets$ Randomly select \textit{Size of test set} points from $D$
    \State Test accuracy of GP model on $T$
\EndFor
\end{algorithmic}
\end{algorithm}

\subsection{GP Models} \label{sec: gp models}
All GP models were implemented with the GPytorch library and utilised the built-in RBF kernel. For binary classification, a custom Laplace approximation for the posterior was implemented, following the approach in \cite{williams2006gaussian}. In this case, hyperparameters were calculated in advance using the built-in GPytorch functionality and then fixed for the duration of the experiments. For multiclass classification, GPytorch's built-in softmax classification function and variational inference were utilised. Hyperparameters were optimised alongside the mean function using the Adam optimiser (500 training iterations and learning rate of 0.05).

\subsection{Baseline Active Learning Strategies}\label{sec: baseline strategies}
This appendix details the baseline active learning strategies used within the experiments. The following strategies were employed:

\begin{enumerate}
    \item \textbf{Random}: Data points are selected randomly from the unlabelled pool. This serves as a simple baseline to compare against more sophisticated strategies.

    \item \textbf{BALD}: BALD selects points that maximise the mutual information between predictions and model parameters. The selected point \( x \) maximises:
    \[
    \text{BALD}(x) = H[y | x, \mathcal{D}] - \mathbb{E}_{p(\theta | \mathcal{D})} \left[ H[y | x, \theta] \right]
    \]
    where \( H \) denotes entropy, \( y \) is the predicted label, \( \theta \) represents model parameters, and \( \mathcal{D} \) is the dataset.

    \item \textbf{Maximum Entropy of the Latent Distribution}: This strategy selects points based on the maximum entropy of the latent distribution \( p(f | x) \). The entropy \( H \) is given by:
    \[
    \text{MaxEnt}(x) = H[p(f | x)]
    \]

    \item \textbf{Entropy of the Predicted Labels}: This method selects data points with the highest entropy in the predicted class probabilities. For a point \( x \), the entropy is:
    \[
    \text{Entropy}(x) = -\sum_{l} p(y = l | x) \log p(y = l | x)
    \]
    where \( p(y = l | x) \) is the predicted probability of label \( l \).

    \item \textbf{Least Confidence}: The least confidence strategy selects points for which the model has the lowest confidence in its most likely prediction. The selection criterion is:
    \[
    \text{LeastConf}(x) = 1 - \max_{l} p(y = l | x)
    \]

    \item \textbf{Margin}: The margin strategy selects points based on the smallest difference between the probabilities of the top two predicted labels. For a point \( x \), it is defined as:
    \[
    \text{Margin}(x) = p(y = l_1 | x) - p(y = l_2 | x)
    \]
    where \( l_1 \) and \( l_2 \) are the classes with the highest and second highest predicted probabilities, respectively.
\end{enumerate}

In the binary case, we refer to the Least Confidence strategy as ``Standard'' because it is equivalent to both the ``Margin'' and ``Entropy of the Predicted Labels'' approaches.

\end{document}